\documentclass[11pt,a4paper,reqno]{amsart}

\usepackage{amssymb}
\usepackage[T1]{fontenc}
\usepackage[latin1]{inputenc}
\usepackage{amsthm}

\usepackage{fancyhdr}

\def \F {{\mathcal F}}

\def \V {{\mathcal V}}

\def \P {{\mathbb P}}
\def \R {{\mathbb R}}

\def \I {{\mathbf 1}}
\def \bF {{\mathbb F}}

\newtheorem{theorem}{Theorem}[section]
\newtheorem{lemma}[theorem]{Lemma}
\newtheorem{definition}[theorem]{Definition}
\newtheorem{remark}[theorem]{Remark}
\newtheorem{proposition}[theorem]{Proposition}

\newtheorem{ass}[theorem]{Assumption}

\newcommand{\ud}{\mathrm d}
\newcommand{\ds}{\displaystyle}
\newcommand{\esp}[2][\mathbb E] {#1\left[#2\right]}
\newcommand{\var}[2][{\rm Var}] {#1\left(#2\right)}
\newcommand{\condespf}[2][\F_t]       {\mathbb E\left.\left[#2\right|#1\right]}
\newcommand{\condespfo}[2][\F_0]       {\mathbb E\left.\left[#2\right|#1\right]}
\newcommand{\condespff}[2][\F_{\tau-}]       {\mathbb E\left.\left[#2\right|#1\right]}
\newcommand{\doleans}[1] {\mathcal E\left(#1\right)}
\newcommand{\condespfww}[2][\F_t^{W^1} \vee \F_t^{W^2}]       {\mathbb E\left.\left[#2\right|#1\right]}
\newcommand{\condespfw}[2][\F_t^W]       {\mathbb E\left.\left[#2\right|#1\right]}

\numberwithin{equation}{section}

\begin{document}

\title{Local Risk-Minimization under the Benchmark Approach}
\author{Francesca Biagini}
\address[Francesca Biagini]{Department of Mathematics, LMU,
Theresienstra{\ss}e, 39, 80333 Munich, Germany.}
\email{biagini@mathematik.uni-muenchen.de}.

\author{Alessandra Cretarola}
\address[Alessandra Cretarola]{Department of Mathematics and Computer Science,
University of Perugia, 
via Vanvitelli, 1, 06123 Perugia, Italy}. \email{alessandra.cretarola@dmi.unipg.it}.

\author{Eckhard Platen}
\address[Eckhard Platen]{University of Technology, Sydney,
Finance Discipline Group and Department of Mathematical
Sciences,
PO Box
123, Broadway, NSW, 2007, Australia.}
\email{Eckhard.Platen@uts.edu.au}.

\date{}

\begin{abstract}
We study the pricing and hedging of derivatives in incomplete financial markets
by considering the local risk-minimization method in the context of the benchmark approach, which will be called benchmarked local risk-minimization.
We show that the proposed benchmarked local risk-minimization allows to handle under extremely weak assumptions a much richer modeling world than the classical methodology.
\end{abstract}

\subjclass[2000]{91B28, 62P05, 62P20.}

\keywords{local risk-minimization, F\"ollmer-Schweizer decomposition, Galtchouck-Kunita-Watanabe decomposition, numéraire portfolio, benchmark approach, real world pricing.}

\maketitle

\section{Introduction}

\noindent The valuation and hedging of derivatives in incomplete financial markets is a frequently
studied problem in mathematical finance. The goal of this paper is to discuss the concept of local risk-minimization under the {\em benchmark approach} (see e.g.~\cite{f02},~\cite{fk08},~\cite{kk},~\cite{p05} and~\cite{ph}), a general modeling framework that only requires the existence of a benchmark, the numéraire portfolio. According to this approach, even under the absence of an equivalent local martingale measure (in short ELMM), contingent claims can be consistently evaluated by means of the so-called {\em real world pricing formula}, which generalizes standard valuation formulas, where the discounting factor is the numéraire portfolio and the pricing measure is the physical probability measure $\P$. Local risk-minimization under the benchmark approach has been also studied in~\cite{kp} in the case of jump-diffusion markets. 
In our paper the approach is more general, since we do not assume any specific market model for the primitive assets whose price processes may have jumps. We analyze the relationship between the classical local risk-minimization and the benchmark approach and revisit this quadratic hedging method in this modeling framework. This is rather different from~\cite{kp}, where the definition of cost process and optimal strategy are revisited in a Brownian setting and the square-integrability condition is dropped. In~\cite{kp} the cost is then interpreted in a different sense in terms of cost condition.
However, we should stress that the concept of risk (see Definition \ref{def:risk}) associated to an admissible strategy is well-defined only if the cost process is assumed to be square-integrable. 
Another difference between the two papers is also that we are considering a general setting with no specification of the asset dynamics, while in~\cite{kp} the Brownian setting of the underlying model plays a crucial role.\\
\indent First of all, we study the local risk-minimization method in the case when the benchmarked asset prices are $\P$-local 
martingales, which will correspond to {\em benchmarked risk-minimization}. This includes continuous market models (see Section \ref{sec:blm}) and a wide class of jump-diffusion models (see for example~\cite{ph}, Chapter 14, pages 513-549).
This property implies several advantages since in market models, where the discounted asset prices are given by $\P$-local martingales, the local risk-minimization method coincides with 
risk-minimization, as introduced originally in~\cite{fs86}. In the local risk-minimization approach, the optimal strategy is often calculated
by switching to a particular martingale measure $\widehat \P$ (the {\em minimal martingale measure}) and computing the Galtchouck-Kunita-Watanabe (in short GKW) decomposition of a benchmarked contingent claim $\hat H$ under $\widehat \P$. However, this method has two main disadvantages:
\begin{itemize}
\item[(i)] the minimal measure $\widehat \P$ may not exist, as it is often the case in the presence of jumps affecting the asset price dynamics; 
\item[(ii)] if $\widehat \P$ exists, the GKW decomposition of $\hat H$ under $\widehat \P$ must satisfy some particular integrability conditions under the real world probability measure $\P$ to give the F\"ollmer-Schweizer 
decomposition of $\hat H$.
\end{itemize} 
On the contrary, the risk-minimization 
approach that we discuss in this paper for the case of benchmarked market models, does not face the same technical difficulties as the local risk-minimization one. It formalizes in a straightforward mathematical way the 
economic intuition of risk and delivers always an optimal strategy for a 
given benchmarked contingent claim $\hat H \in L^2(\F_T,\P)$\footnote{The space $L^2(\F_T,\P)$ denotes the set of all $\F_T$-measurable random variables $H$ such that $\esp{H^2} = \int H^2\ud \P < \infty$.}, obtained by computing the GKW decomposition of $\hat H$ under $\P$. \\
\indent Furthermore, 
in this setting we establish a fundamental relationship between real world pricing and benchmarked risk-minimization. In market models, where the asset prices are given by $\P$-local martingales, by Theorem \ref{prop:fs} we will obtain the result that the benchmarked portfolio's value of the risk-minimizing strategy for $\hat H \in L^2(\F_T,\P)$ coincides with the real world pricing formula for $\hat H$. The benchmarked contingent claim $\hat H$ can be written as
\begin{equation} \label{Hdecomposition}
\hat H=\hat H_0 + \int_0^T \xi_u^{\hat H} \ud \hat S_u + L_T^{\hat H}\quad \P-{\rm a.s.},
\end{equation}
where $L^{\hat H}$ is a square-integrable $\P$-martingale with $L_0^{\hat H}=0$ strongly orthogonal\footnote{Two $\P$-local martingales $M$ and $N$ are called {\em strongly orthogonal} if their product $MN$ is a $\P$-local martingale.} to $\hat S$.
Decomposition \eqref{Hdecomposition} allows us to decompose every 
square-integrable benchmarked contingent claim as the sum of its {\em hedgeable
part} $\hat H^h$ and its
{\em unhedgeable part} $\hat H^u$ such that we can write
\begin{equation*} 
\hat H=\hat H^h+\hat H^u,
\end{equation*}
where
\begin{equation*} 
\hat H^h:=\hat H_0+\int_0^T\xi_u^{\hat H}\cdot \ud \hat
S_u
\end{equation*}
and
\begin{equation*} 
\hat H^u:=L_T^{\hat H}.
\end{equation*}
Here the notation $\int \xi^{\hat H} \cdot \ud \hat S$
characterizes the integral of the vector process $\xi^{\hat H}$ with respect
to the vector process $\hat S$ (see e.g.~\cite{mr}). 
Note that the benchmarked hedgeable part $\hat H^h$ can be replicated perfectly, i.e.
\begin{equation*} 
\hat U_{H^h}(t)=\condespf{\hat H^h}=\hat
H_0+\int_0^t\xi^{\hat H}_u\cdot \ud \hat S_u\,,
\end{equation*}
and $\xi^{\hat H}$ yields the fair strategy
for the self-financing replication of the hedgeable part of $\hat H$. The remaining benchmarked 
unhedgeable part can be diversified and will be covered through the cost process $C:=L^{\hat H}-\hat H_0$.
The connection between risk-minimization and real world pricing 
is then an important insight, which gives a clear reasoning for the pricing and hedging of contingent claims via real world pricing also in incomplete markets.\\
\indent A natural question concerns indeed the invariance of the risk-minimizing strategy under a change of numéraire. By~\cite{bp} this property always holds in the case of continuous assets prices. Here we show that this result is also true more generally: it is sufficient that
the orthogonal martingale structure is generated by continuous $\P$-(local) martingales. \\
\indent Then we also study the case when the benchmarked processes are 
$\P$-supermar-\linebreak tingales. 
In particular we consider a market model, where incompleteness is due to incomplete information. In this setting we show that a benchmarked locally risk-minimizing strategy can be determined by computing the predictable projection of the strategy in the completed market without any specification of the asset price dynamics (see Theorem \ref{th:incomplete}).
The proof we provide holds when the discounted asset prices are special semimartingales in $\mathcal S^2(\P)$\footnote{Given the Doob-Meyer decomposition
$$
X_t=X_0+M_t+V_t, \quad t \in [0,T],
$$
of a $\P$-semimartingale $X$ into a $\P$-local martingale $M=\{M_t, \ t\in [0,T]\}$ and an $\bF$-predictable process $V=\{V_t, \ t\in [0,T]\}$ of finite variation, we say that $X \in \mathcal S^2(\P)$ if the following integrability condition is satisfied
$$
\esp{X_0^2 + [X]_T + |V|_T^2}<\infty.
$$
Here $|V|=\{|V|_t, \ t\in [0,T]\}$ denotes the total variation of the process $V$.}, hence in particular for all benchmarked underlying assets in $\mathcal S^2(\P)$ by the Doob's decomposition. This extends the results of~\cite{fs}, where they assume continuity of the underlying prices processes. \\
\indent Finally, we provide some examples to illustrate how to compute the GKW 
decomposition in the minimal market model with random scaling, where there exists no ELMM, but the primitive assets are still $\P$-local martingales if benchmarked. \\
\indent The local risk-minimization method under the benchmark approach has acquired new importance for pricing and hedging in hybrid markets and insurance markets (see~\cite{b1} and~\cite{bw}). Since hybrid markets are intrinsically incomplete, perfect replication of contingent claims is not always possible and one has to apply one of the several methods for pricing and hedging in incomplete markets. Local risk-minimization appears to be one of the most suitable methods when the market is affected by orthogonal sources of randomness, such as the ones represented by mortality risk and catastrophic risks. The results of this paper provide a new simplified framework for applying benchmarked local risk-minimization.

\section{Financial Market} \label{finmarkt}

To describe a financial market in continuous time,
we introduce a probability space $(\Omega,\F,\P)$, a time horizon
$T \in (0,\infty)$ and a filtration $\bF:=(\F_t)_{0\leq t\leq T}$
that is assumed to satisfy $\F_t \subseteq \F$ for all $t \in
[0,T]$, as well as the usual hypotheses of completeness and
right-continuity and saturation by
all $\P$-null sets of $\F$.\\
In our market model we can find $d$ adapted, nonnegative primary
security account processes represented by (càdlàg) $\P$-semimartingales $S^j=\{S_t^j,\ t \in [0,T]\}$, $j \in \{1,2,\ldots,d\}$,
$d\in \{1,2,\ldots\}$. Additionally, the $0$-th security
account $S_t^0$ denotes the value of the adapted strictly positive 
savings account at time $t \in [0,T]$.
The $j$-th {\em primary security account} holds units of the $j$-th primary security plus its accumulated
dividends or interest payments, $j \in
\{1,2,\ldots, d\}$. 
In this setting, market participants can trade in order to reallocate
their wealth.
\begin{definition}
A {\em strategy} is a $(d+1)$-dimensional process
$\delta=\{\delta_t=(\delta_t^0,\delta_t^1,\ldots,\delta_t^d)^\top$,
$t\in [0,T]\}$,  
where for each $j \in
\{0,1,\ldots,d\}$, the process $\delta^j=\{\delta_t^j,\ t \in[0,T]\}$
is $\bF$-predictable and integrable with respect to $S^j=\{S_t^j,\ t
\in [0,T]\}$. 
\end{definition}
\noindent Here
$\delta_t^j$, $j \in
\{0,1,\ldots,d\}$, denotes the number of units of the $j$-th
security account that are held at time $t \ge 0$ in the
corresponding
{\em portfolio} $S^\delta=\{S_t^\delta,\ t \in
[0,T]\}$. Following~\cite{bp}, we define the 
value $S^\delta$ of this portfolio as given by a càdlàg optional process such that 
$$
S_{t-}^\delta:=\delta_t \cdot S_t=\sum_{j=0}^d \delta_t^j S_t^j, \quad t \in [0,T],
$$
where $S=\{S_t=(S_t^0,S_t^1,\ldots,S_t^d)^\top,\ t
\in [0,T]\}$.
A strategy $\delta$ and the corresponding
portfolio $\hat S^\delta$ are
said to be {\em self-financing} if
\begin{equation} \label{eq:self}
S_t^\delta=S_0^\delta+\int_0^t\delta_u \cdot \ud S_u, \quad t \in [0,T],
\end{equation}
where
$\delta=\{\delta_t=(\delta_t^0,\delta_t^1,\ldots,\delta_t^d)^\top,\
t\in [0,T]\}$. 
Note that the stochastic integral of the vector process $\delta$ with respect
to $S$ is well-defined because of our assumptions on $\delta$. Furthermore, $a^\top$ denotes the transpose of $a$. In general, we do not request strategies to be
self-financing. Denote by $\V_x^+$, ($\V_x$), the set of
all strictly positive, (nonnegative), finite, 
self-financing 
portfolios, with initial
capital $x >0$, ($x \ge 0$).
\begin{definition} \label{def:numport}
A portfolio
$S^{\delta_*} \in \V_1^+$ is called a {\em numéraire portfolio}, if 
any nonnegative portfolio $S^\delta \in \V_1^+$, when denominated in units of $S^{\delta_*}$,  
 forms a $\P$-supermartingale, that is,
\begin{equation}\label{supermprop}
\frac{S_t^\delta}{S_t^{\delta_*}} \ge \condespf{\frac{S_s^\delta}{S_s^{\delta_*}}},
\end{equation}
for all $0 \leq t \leq s\leq T$.
\end{definition}
\noindent 
To establish the modeling framework, we make the following
(extremely weak) key assumption, which is satisfied for almost all
models of practical interest, see e.g.~\cite{ph} and~\cite{kk}.
\begin{ass}\label{ass:numport}
There exists a
numéraire portfolio $S^{\delta_*}
\in \V_1^+$.
\end{ass}
\noindent From now on, let us choose the numéraire portfolio as {\em benchmark}. We call any security, when expressed in units of the numéraire portfolio, a benchmarked security and refer to this procedure as {\em benchmarking}. The benchmarked value of a portfolio $S^\delta$ is of particular interest and is given by the ratio
$$
\hat S_t^\delta=\frac{S_t^\delta}{S_t^{\delta_*}}
$$
for all $t \in [0,T]$. If a benchmarked price process is a $\P$-martingale, then we call it {\em fair}. In this case we would have equality in relationship \eqref{supermprop} of Definition \ref{def:numport}.\\
The benchmark approach developed in~\cite{kk},~\cite{l90} and~\cite{ph} uses the numéraire portfolio for derivative pricing without using equivalent martingale measures.
In portfolio optimization the numéraire portfolio, which is also the growth optimal portfolio, is in many other ways
the best performing self-financing portfolio, see~\cite{k56} and~\cite{m73}.\\
As shown in~\cite{ph}, jump-diffusion and It\^o process driven market models have a numéraire portfolio under very general assumptions, where benchmarked nonnegative portfolios turn out to be $\P$-local martingales and, thus, $\P$-supermartingales. In~\cite{kk} the question on the existence of a numéraire portfolio in a general semimartingale market is studied. \\
In order to guarantee the economic viability of
our framework, we check whether obvious
arbitrage opportunities are excluded.
A strong form of arbitrage would arise when a market
participant could generate strictly positive wealth from zero
initial capital via his or her nonnegative portfolio of total wealth.
\begin{definition}
A benchmarked nonnegative self-financing 
portfolio $\hat S^\delta$ 
is a {\em strong arbitrage} if
it starts with zero initial capital, that is $\hat S_0^\delta=0$, and
generates some strictly positive wealth with strictly positive
probability at a later time $t \in (0,T]$, that is
$\P(\hat S_t^\delta>0)>0$.
\end{definition}
\noindent Thanks to the supermartingale property \eqref{supermprop}, the existence of
the numéraire portfolio guarantees that strong arbitrage is
automatically excluded in the given general setting, see~\cite{ph}. However, some weaker forms of arbitrage may still exist. These would require to allow for negative portfolios of total wealth of those market participants who fully focus on exploiting such weaker forms of arbitrage, which is not possible in reality due to bankruptcy laws. This emphasizes the fact that an economically motivated notion of arbitrage should rely on nonnegative portfolios.\\
Within this paper, we consider a discounted European style contingent claim. Such a benchmarked claim $\hat H$ (expressed in
units of the benchmark) is given by the $\F_T$-measurable,
nonnegative random payoff $\hat H$
that is delivered at time $T$. We will here always assume that a benchmarked contingent claim
$\hat H$ belongs to
$L^2(\F_T,\P)$.\\
\noindent Given a benchmarked contingent claim $\hat H$, there are at least two
tasks that a potential seller of $\hat H$ may want to accomplish: the {\em pricing}
by assigning a value to $\hat H$ at times $t < T$; and the {\em hedging} by
covering as much as possible against potential losses arising from the uncertainty
of $\hat H$. If the market is complete, then there exists a
self-financing strategy $\delta$ whose terminal value $\hat S_T^\delta$
equals $\hat H$ with probability one, see~\cite{ph}.
More precisely, the {\em real world pricing formula}
\begin{equation} \label{realwpf}
\hat S_t^{\delta_H}=\condespf{\hat H}
\end{equation}
provides the description for the benchmarked fair portfolio at time $t \in [0,T]$, which is the least expensive $\P$-supermartingale that replicates the benchmarked payoff $\hat H$ if it admits a replicating self-financing
strategy $\delta^H$ with $\hat S_T^{\delta^H}=\hat H$. Here
$\hat S^{\delta_H}$ forms by definition a $\P$-martingale.
The benchmark approach allows other self-financing  hedge portfolios to exist for $\hat H$, see~\cite{ph}. However, these nonnegative portfolios are not $\P$-martingales and, as $\P$-supermartingales, more expensive than the $\P$-martingale $\hat S^{\delta_H}$ given in \eqref{realwpf}, see~\cite{ph}. \\
Completeness is a rather
delicate property that does not cover a large class of realistic
market models. Here we choose the (local) risk-minimization approach (see e.g.~\cite{fs},~\cite{fs86} and~\cite{s01}) to price non-hedgeable contingent claims. 

\noindent In this paper, we first  
investigate the case of benchmarked securities that represent $\P$-local martingales and study risk-minimization as originally introduced in~\cite{fs86}. 
We will see that this covers many cases in the context of the benchmark approach including all continuous financial market models,
a wide range of jump-diffusion
driven market models and cases like the minimal market model that do not have an equivalent risk neutral probability measure. Then we will study the general case when benchmarked securities are $\P$-supermartingales that are not necessarily $\P$-local martingales. As indicated earlier, we will refer to local risk-minimization under the benchmark approach as {\em  benchmarked local risk-minimization}.

\section{Local Risk-Minimization with Benchmarked Assets} \label{lrm}

Our aim is to investigate a concept of local risk-minimization similar to the one in~\cite{hps01} and~\cite{s01}, which used the savings account as reference unit. Here we use the numéraire portfolio as discounting factor and benchmark. 
The main feature of a local risk-minimization concept is the
fact that one insists on the replication requirement $\hat S_T^\delta=\hat H$. If $\hat H$
is not hedgeable, then this forces one to work with strategies that are
not self-financing and the aim becomes to minimize the resulting intrinsic risk or cost under
a suitable criterion. As we will see, rather natural and tractable are quadratic
hedging criteria, where we refer to~\cite{s01} and~\cite{hps01} for extensive surveys.\\
Important is the fact that there are realistic situations that we will cover, which would be excluded because a minimal martingale measure may not exist for the respective models. For example, in the case of local risk-minimization of financial derivatives based on insurance products, the minimal martingale measure may often not exist because of the presence of jumps in the underlying. \\
We recall that under Assumption \ref{ass:numport}, the benchmarked
value of any nonnegative, self-financing portfolio forms a $\P$-supermartingale, see \eqref{supermprop}. In
particular, the vector of the $d+1$ benchmarked primary security accounts
$\hat S=(\hat S^0,\hat S^1,\ldots, \hat S^d)^\top$ forms with each of its components a nonnegative
$\P$-{\it supermartingale}. By Theorem VII.12 of~\cite{dm2}, we know that the
vector process $\hat S$ has a unique decomposition of the form
\begin{equation} \label{supermartdeco}
\hat S_t=\hat S_0+M_t+V_t,\quad t \in [0,T],
\end{equation}
where $M$ is a vector $\P$-local martingale and $V$ is a right-continuous
$\bF$-predictable finite variation 
vector process with $M_0=V_0={\bf 0}$, with ${\bf 0}$ denoting the
$(d+1)$-dimensional null vector. This expresses the fact that every right-continuous $\P$-supermartingale is a special $\P$-semimartingale. 

\subsection{Benchmarked Local Martingales} \label{sec:blm}

We now discuss the case when benchmarked securities are $\P$-local martingales.
Let us assume that the vector of the $d+1$ discounted
primary security accounts $\displaystyle \frac{S}{S^0}=:X=\{X_t=(1,X_t^1,\ldots,X_t^d)^\top,\ t \in [0,T]\}$ 
is a {\em continuous} $\P$-semimartingale with canonical decomposition $X=X_0+M^X+A^X$. The processes $M^X=\{M_t^{X}:\ t\in[0,T]\}$ and $A^X=\{A_t^{X}:\ t\in[0,T]\}$ are both $\R^{d+1}$-valued, continuous and null at $0$. Moreover, $M^X$ is a vector $\P$-local martingale and $A^X$ is an
adapted, finite variation vector process. The bracket process $\langle M^X\rangle$ of $M^X$ is the adapted, continuous $(d+1) \times (d+1)$-matrix-valued process with components $\langle M^X\rangle_t^{i,j}=\langle (M^X)^i, (M^X)^j\rangle_t$ denoting covariation for $i,j=0,1, \ldots, d$ and $t \in [0,T]$.\\
Since Assumption \ref{ass:numport} is in force, Theorem 3.4 of~\cite{hs10} ensures that the {\em structure condition}\footnote{We say that $X$ satisfies the {\em structure condition} if $A^X$ is absolutely continuous with respect to $\langle M^X \rangle$, in the sense that there exists an $\bF$-predictable process $\hat \lambda=\{\hat \lambda_t, t \in [0,T]\}$ such that $A^X= \int \ud \langle M^X \rangle \hat \lambda$, i.e.
$
(A_t^X)^{i}=\sum_{j=0}^d \int_0^t \hat \lambda_u^j \ud \langle M^X\rangle_u^{ij}$, for $i \in \{0,\ldots,d\}$ and $t \in [0,T]$,  
and the {\em mean-variance tradeoff process}
$
\hat K_t=\int_0^t \hat \lambda_u^\top \ud \langle M^X\rangle_u\hat \lambda_u
$
is finite $\P$-a.s. for each $t \in [0,T]$.} is satisfied and the discounted
numéraire portfolio $\bar S_t^{\delta_*}=\frac{S_t^{\delta_*}}{S_t^0}$ at any time $t$ is given by
$$
\bar S_t^{\delta_*}=\frac{1}{\hat Z_t},\quad t \in [0,T],
$$
where the process $\hat Z$ corresponds to the stochastic exponential
$$
\hat Z_t=\doleans{-\hat \lambda \cdot M^X}_t= \exp \left(-\hat \lambda \cdot M_t^X - \frac{1}{2} \hat K_t\right), \quad t \in [0,T],
$$
which is then well-defined and a strictly positive $\P$-local martingale. Via It\^o's product rule, it is easy to check that the vector process $\hat S$ of benchmarked primary security accounts is a $\P$-local martingale, and thus, a $\P$-supermartingale. Indeed, since $X_t=X_0+M_t^X+\int_0^t \hat \lambda_s \ud \langle M^X\rangle_s$, we have
\begin{align*}
\ud \hat S_t & = \ud (X_t \hat Z_t) = \hat Z_t \ud X_t + X_t \ud \hat Z_t + \ud \langle X,\hat Z \rangle_t\\
& = \hat Z_t(1-X_t\hat \lambda_t) \ud M_t^X, \quad t \in [0,T].
\end{align*} 
This implies that whenever we consider continuous primary security account processes, they are $\P$-local martingales when expressed in units of the numéraire portfolio.\\
In the general case when $S_t$ can have jumps, it is not possible to provide an analogous explicit description of the numéraire portfolio $S_t^{\delta_*}$ or, more precisely, its generating strategy $\delta_*$. An implicit description can be found in~\cite{kk}, Theorem 3.15, or more generally in~\cite{gk}, Theorem 3.2 and Corollary 3.2. In both cases, $\delta_*$ can be obtained by pointwise maximization of a function that is given explicitly in terms of semimartingale characteristics. If $S$ is discontinuous, such a pointwise maximizer is only defined implicitly and neither of the above descriptions provides explicit expressions for $\delta_*$.\\
However, a wide class of jump-diffusion market models is driven by primary security account processes that turn out to be, when expressed in units of the numéraire portfolio, $\P$-local martingales, see e.g.~\cite{ph}, Chapter 14.
For example, this is the case in jump-diffusion markets, that is, when security price processes exhibit intensity based jumps due to event risk, see~\cite{ph}, Chapter 14, page 513. These results allow us to consider below  risk-minimization in the case when the benchmarked assets are given by $\P$-local martingales.
 
\subsubsection{Risk-Minimization with Benchmarked Assets}

\noindent Since at this stage we refer to the case where benchmarked securities represent $\P$-local martingales (i.e. we assume $V \equiv {\bf 0}$ for all $t \in [0,T]$ in \eqref{supermartdeco}), we study risk-minimization as originally introduced in~\cite{fs86} under the benchmark approach, that is, benchmarked risk-minimization. In particular, since we are considering a (general) discounting factor (different from the usual money market account), we follow the approach of~\cite{bp} for local risk-minimization under a given numéraire.
\begin{definition}
An $L^2$-{\em admissible strategy} is any $\R^{d+1}$-valued $\bF$-predictable vector process
$\delta=\{\delta_t=(\delta_t^0,\delta_t^1,\ldots,$ $\delta_t^d)^\top,\
t\in [0,T]\}$ 
such that
\begin{enumerate}
\item[(i)] the associated portfolio $\hat S^\delta$ is a square-integrable stochastic process whose left-limit is equal to $\hat S^\delta_{t-}=\delta_t \cdot \hat S_t$, 
\item[(ii)] the stochastic integral $\int \delta \cdot \ud \hat S$ 
is such that
\begin{equation}\label{admissible}
\esp{\int_0^T\delta_u^\top\ud [\hat S]_u \delta_u}<\infty.
\end{equation}
Here $[\hat S]=([\hat S^i,\hat S^j])_{i,j=1,\ldots,d}$ denotes the matrix-valued optional covariance process of $\hat S$.
\end{enumerate}
\end{definition}
\noindent Recall that the market may be not complete. We also admit strategies 
that are not self-financing and may generate benchmarked profits or
losses over time.
\begin{definition} For any $L^2$-admissible strategy
$\delta$, the {\em benchmarked cost process} $\hat
C^\delta$ is defined by
\begin{equation}\label{profloss}
\hat C_t^\delta:=\hat S_t^\delta-\int_0^t \delta_u \cdot \ud
\hat S_u
,\quad t \in [0,T].
\end{equation}
\end{definition}
\noindent Here $\hat C_t^\delta$ describes the total costs incurred by $\delta$ over the interval $[0,t]$. 
\begin{definition} \label{def:risk}
For an $L^2$-admissible strategy $\delta$, the corresponding {\em
risk} at time $t$ is defined by
$$
\hat R_t^\delta:=\condespf{\left(\hat C_T^\delta-\hat C_t^\delta
\right)^2}, \quad t \in [0,T],
$$
where the benchmarked cost process $\hat C^\delta$, given
in \eqref{profloss}, is assumed to be square-integrable.
\end{definition}
\noindent If $\hat C^\delta$ is constant, then it equals zero and the strategy is self-financing. 
Our goal is to find an $L^2$-admissible strategy
$\delta$, which minimizes the associated risk measured by the
fluctuations of its benchmarked cost process in a suitable sense.
\begin{definition} \label{optimalstrategy}
Given a benchmarked contingent claim $\hat H \in L^2(\F_T,\P)$,
an $L^2$-admissible strategy $\delta$ is said to be {\em
benchmarked
risk-minimizing}
if the following conditions hold:
\begin{itemize}
\item[(i)] $\hat S_T^\delta=\hat H,\ \P$-a.s.;
\item[(ii)] for any $L^2$-admissible strategy $\tilde \delta$ such that $\hat S_T^{\tilde \delta}=\hat S_T^\delta$ $\P$-a.s., we have
$$
\hat R_t^\delta \leq \hat R_t^{\tilde \delta} \quad \P-\mbox{a.s.}\ \mbox{for}\ \mbox{every}\ t \in [0,T].
$$
\end{itemize}
\end{definition}
\begin{lemma} \label{lem:msf}
The benchmarked cost process $\hat C^\delta$ defined in \eqref{profloss} associated to a benchmarked 
risk-minimizing strategy $\delta$ is a $\P$-martingale for all $t \in [0,T]$.
\end{lemma}
\noindent For the proof of Lemma \ref{lem:msf}, we refer to Section \ref{tecpr} in the Appendix.
Hence benchmarked risk-minimizing strategies are ``self-financing on average''.
We will see, to find a benchmarked 
risk-minimizing strategy corresponds
to finding 
a suitable decomposition of the benchmarked claim adapted to this setting. Let $\mathcal M_0^2(\P)$
be the space of all square-integrable $\P$-martingales starting at null at the initial time. 
\begin{theorem} \label{prop:fs}
Every benchmarked contingent claim $\hat H \in L^2(\F_T,\P)$ admits a unique benchmarked risk-minimizing strategy $\delta$ with portfolio value $\hat S^\delta$ and benchmarked cost process $\hat C^\delta$, given by
\begin{align*}
\delta & = \delta^{\hat H},\\
\hat S_t^\delta & = \hat H_t=\condespf{\hat H}, \quad t \in [0,T],\\
\hat C^\delta & = \hat H_0
+L^{\hat H},
\end{align*}
where $\delta^{\hat H}$ and $L^{\hat H}$ are provided by the Galtchouk-Kunita-Watanabe decomposition of $\hat H$, i.e.
\begin{equation} \label{fsdecomp}
\hat H=\hat H_0+\int_0^T\delta_u^{\hat H}\cdot \ud \hat S_u+L_T^{\hat H},\quad \P-{\rm a.s.}
\end{equation}
with $\hat H_0\in \R$, where $\delta^{\hat H}$ is an $\bF$-predictable vector process  satisfying the integrability condition \eqref{admissible} and $L^{\hat H} \in \mathcal M_0^2(\P)$ is strongly orthogonal to each component of $\hat S$.
\end{theorem}
\begin{proof}
 The proof follows from Theorem 2.4 of~\cite{s01} and Lemma \ref{lem:msf}.
\end{proof}
\noindent 
Thus, the problem of minimizing risk is reduced to finding the representation \eqref{fsdecomp}.
\noindent A natural question is whether the benchmarked risk-minimizing strategy is invariant under a change of numéraire. We address this issue to Section \ref{sec:inv}. 

\subsection{Relationship to Real World Pricing}

\begin{definition} \label{def:hedge}
We say that a nonnegative benchmarked contingent claim $\hat H \in
L^2(\F_T,\P)$ is {\em hedgeable} if there exists an
$L^2$-admissible self-financing strategy
$\xi^{\hat H}=\{\xi_t^{\hat H}=(\xi^{\hat H,1}_t,\ldots,$      
$\xi^{\hat H,d}_t)^\top,$ $\ t \in [0,T]\}$ such that
\begin{equation*} \label{hedgerepr}
\hat H=\hat H_0+\int_0^T\xi_u^{\hat H}\cdot \ud \hat S_u.
\end{equation*}
\end{definition}
\noindent 
Decomposition \eqref{fsdecomp} and Definition
\ref{def:hedge} allow us to decompose every nonnegative,
square-integrable benchmarked contingent claim as the sum of its {\em hedgeable
part} $\hat H^h$ and its
{\em unhedgeable part} $\hat H^u$ such that we can write
\begin{equation} \label{eq:hdecomp}
\hat H=\hat H^h+\hat H^u,
\end{equation}
where
\begin{equation*} \label{eq:hhedge}
\hat H^h:=\hat H_0+\int_0^T\xi_u^{\hat H}\cdot \ud \hat
S_u
\end{equation*}
and
\begin{equation*} \label{eq:hunhedge}
\hat H^u:=L_T^{\hat H}.
\end{equation*}
Recall that $L^{\hat H}=\{L_t^{\hat H},t \in [0,T]\}$ is a $\P$-martingale in $\mathcal M_0^2(\P)$, strongly
orthogonal to each component of $\hat S$.
\noindent There is a close relationship between benchmarked risk-minimization and real world pricing, as we will see now.
Let us apply the real world pricing formula \eqref{realwpf}
to the benchmarked contingent claim
$\hat H$ in order to get its benchmarked
fair price
$\hat U_{H}(t)$ at time $t$. Recall, by its martingale property that the benchmarked fair price is the best forecast of its future
benchmarked prices. Due to the supermartingale property \eqref{supermprop}, it  follows that we characterize, when using the real world pricing formula \eqref{realwpf} for obtaining
the fair price of the hedgeable part, the
least expensive replicating portfolio for $\hat H^h$ by taking
the conditional
expectation $\condespf{\hat H^h}$ under the real world probability
measure $\P$. Then by \eqref{eq:hdecomp} we have
\begin{equation*} \label{uht}
\hat U_{H}(t)=\condespf{\hat
H}=\condespf{\hat H^h}+\condespf{\hat H^u}=\hat U_{H^h}(t)+\hat
U_{H^u}(t),
\end{equation*}
for every $t \in
[0,T]$.
Note that the benchmarked hedgeable part $\hat H^h$ can be replicated perfectly, i.e.
\begin{equation*} \label{uhh}
\hat U_{H^h}(t)=\condespf{\hat H^h}=\hat
H_0+\int_0^t\xi^{\hat H}_u\cdot \ud \hat S_u.
\end{equation*}
In particular, for $t=0$ one has for the benchmarked hedgeable part
$$
\hat U_{H^h}(0)=\condespfo{\hat H^h}=\hat H_0.
$$
On the other hand, we have for the benchmarked unhedgeable part
$$
\hat U_{H^u}(t)=\condespf{\hat H^u}=L_t^{\hat H}
$$
with
$$
\hat U_{H^u}(0)=0.
$$
Consequently, for the nonnegative benchmarked payoff $\hat H$, its benchmarked fair price
$\hat U_{H}(0)$ at time $t=0$, is given by
$$
\hat U_{H}(0)=\hat U_{H^h}(0)+\hat U_{H^u}(0)=\condespfo{\hat
H^h}+\condespfo{\hat H^u}=\hat H_0.
$$
The real world pricing formula \eqref{realwpf} appears in the form of a conditional expectation
and, thus, as a projection in a least squares sense. More
precisely, the benchmarked fair price $\hat U_{H}(0)$ can be
interpreted as the least squares projection of $\hat H$ into the space of
$\F_0$-measurable benchmarked values. Note that the
benchmarked fair price $\hat U_{H^u}(0)$ of the unhedgeable part $\hat H^u=L_T^{\hat H}$
is zero at time $t=0$. Recall that the benchmarked hedgeable part is priced at time $t=0$ such that the minimal possible price, the fair price, results.
Viewed from time $t=0$ the benchmarked cost $\hat C_T^\delta=\hat H_0+\hat L_T^{\hat H}$, see Theorem \ref{prop:fs},
has then 
minimal variance $\var{L_T^{\hat H}}$.
This means that the application of the real world pricing formula to a benchmarked payoff at time $t=0$
leaves its benchmarked unhedgeable part {\em totally untouched}. This is
reasonable because any extra trading could only create unnecessary
uncertainty and potential additional benchmarked costs. Of course, once the benchmarked fair price is used to establish a hedge portfolio, a benchmarked cost emerges according to Theorem \ref{prop:fs} if there was an unhedgeable part in the benchmarked contingent claim.
The following practically important insight is worth mentioning:
\begin{remark}
From a large financial institution's point of view,
the benchmarked profits {\em \&} losses due to the optimal costs in its derivative book have 
minimal variance when evaluated under real world pricing and viewed at time $t=0$. If they are large in number and {\em independent}, then the Law of Large Numbers
reduces asymptotically the variance of the benchmarked pooled profit {\em \&} loss to zero and, thus, its value to zero. \\
Obviously, requesting from clients higher prices
than fair prices would make the bank less competitive. On the other hand, charging lower prices than fair prices would make
it unsustainable in the long run because it would suffer on average a loss. In this sense fair pricing of unhedgeable
claims is most natural and yields economically correct prices. Accordingly, benchmarked
risk-minimization is a
very natural risk management strategy for pricing and hedging. Moreover, it is mathematically convenient and for many models rather tractable when using the GKW-decomposition.
\end{remark}
\noindent With the above notation, we obtain by Theorem \ref{prop:fs} 
and 
\eqref{fsdecomp}
for the benchmarked payoff $\hat H \in L^2(\F_T,\P)$  the following decomposition:
\begin{equation*}
\hat H=\hat U_{H^h}(0)+\int_0^T\xi^{\hat H}_u\cdot \ud \hat
S_u+\hat U_{H^u}(T).
\end{equation*}
Since $\hat U_{H^u}(T)=\hat H^u=L_T^{\hat H}$, it follows
\begin{equation*} \label{eq:realworld}
\hat H=\hat U_{H^h}(0)+\int_0^T\xi^{\hat H}_u\cdot \ud \hat
S_u+L_T^{\hat H}.
\end{equation*}
This allows us to summarize the relationship between benchmarked 
risk-minimiza-\linebreak tion and real world pricing.
In our setting $\hat H \in L^2(\F_T,\P)$ admits a benchmarked risk-minimizing strategy and the decomposition for $\hat H$, provided by the real world
pricing formula, coincides with the decomposition \eqref{fsdecomp},
where $\xi^{\hat H}$ yields the fair strategy
for the self-financing replication of the hedgeable part of $\hat H$.
The remaining benchmarked unhedgeable part, given by the benchmarked cost process $ L^{\hat H}$, can be diversified. Note that diversification takes place under the real world probability measure and {\it not} under some putative risk neutral measure. This is an important insight, which gives a clear reasoning for the pricing and hedging of contingent claims via real world pricing in incomplete markets.

\subsection{Local Risk-Minimization with Benchmarked Assets} \label{ss:lrm} 

We now consider the general situation, where the vector of the $d+1$ benchmarked primary security accounts $\hat S=(\hat S^0,\hat S^1,\ldots,\hat S^d)^\top$ forms with each of its components a locally square-integrable nonnegative $\P$-supermartingale with $V \neq 0$ in decomposition \eqref{supermartdeco}, and hence a special $\P$-semimartingale. In view of Proposition 3.1 in~\cite{s01}, Definition \ref{def:risk} does not hold in this non-martingale case due to a compatibility problem. Indeed as observed in~\cite{s01}, at any time $t$ we minimize $\hat R_t^\delta$ over all admissible continuations from $t$ on and obtain a continuation which is optimal when viewed in $t$ only. But for $s<t$, the $s$-optimal continuation from $s$ onward highlights what to do on the whole interval $(s,T] \supset (t,T]$ and this may be different from what the $t$-optimal continuation from $t$ on prescribes. However, it is possible to characterize benchmarked pseudo-locally risk-minimizing strategies\footnote{The original definition of a {\it locally risk-minimizing} strategy is given in~\cite{s01} and formalizes the intuitive idea that changing an optimal
strategy over a small time interval increases the risk, at least asymptotically. Since it is a rather technical definition, it has been introduced the concept of a {\it pseudo-locally risk-minimizing} strategy that is
both easier to find and to characterize, as Proposition \ref{pr:mainres} will show in the following. Moreover, in the one-dimensional case and if $\hat S$ is sufficiently well-behaved, pseudo-optimal and
locally risk-minimizing strategies are the same.} through the following well-known result, see~\cite{s01}.  
\begin{proposition}\label{pr:mainres}
A benchmarked contingent claim $\hat H \in L^2(\F_T,\P)$ admits a benchmarked pseudo-locally risk-minimizing strategy $\delta$ with $\hat S_T^\delta=\hat H$ $\P$-a.s. if and only if $\hat H$ can be written as
\begin{equation}\label{fsdecomp1}
\hat H=\hat H_0+\int_0^T\xi_u^{\hat H}\cdot \ud \hat S_u+L_T^{\hat H}, \quad \P-{\rm a.s.}
\end{equation}
with $\hat H_0 \in L^2(\F_0,\P)$, $\xi^{\hat H}$ is an $\bF$-predictable vector process satisfying the following integrability condition
$$
\esp{\int_0^T(\xi_s^{\hat H})^\top \ud [M]_s \xi_s^{\hat H}+\left(\int_0^T|(\xi_s^{\hat H})^\top||\ud V_s|\right)^2}<\infty,
$$
where for $\omega \in \Omega$, $\ud V_s(\omega)$ denotes the (signed) Lebesgue-Stieltjes measure corresponding to the finite variation function $s \mapsto V_s(\omega)$ and $|\ud V_s|(\omega)$ the associated total variation measure, and $L^{\hat H} \in \mathcal M_0^2(\P)$ is strongly orthogonal to $M$. The strategy $\delta$ is then given by
$$
\delta_t=\xi_t^{\hat H}, \quad t \in [0,T],
$$
its benchmarked value process is 
$$
\hat S_t^\delta=\hat S_0^\delta+\int_0^t \delta_s\cdot\ud \hat S_s+\hat C_t^\delta, \quad t \in [0,T],
$$
and the benchmarked cost process equals
$$
\hat C_t^\delta=\hat H_0
+ \hat L_t^{\hat H}, \quad t \in [0,T].
$$
\end{proposition}
\noindent Decompositions \eqref{fsdecomp} and \eqref{fsdecomp1} for $\hat H \in L^2(\F_T,\P)$ are also known in the literature as the F\"ollmer-Schweizer decompositions for $\hat H$.

\subsubsection{Invariance under a Change of Numéraire} \label{sec:inv}

A natural question to clarify is when risk-minimizing strategies are invariant under a change of numéraire.
Indeed, if the primary security accounts $S^j$, $j \in \{0,1,\ldots,d\}$, are continuous, then Theorem
3.1 of~\cite{bp} ensures that the strategy is invariant under a
change of numéraire. In this case the process $\xi^{\hat H}$ appearing in decomposition \eqref{fsdecomp1}
also provides the classical locally risk-minimizing strategy (if it exists) for the discounted contingent claim $\bar H:=\ds \frac{H}{S_T^0}$.
However, if the primary security accounts $S^j$ are only  right-continuous, it is still possible to extend some results of~\cite{bp} as follows:\\
Consider two discounting factors $S^0$ and $S^{\delta_*}$. 
Given an $L^2$-admissible strategy
$\delta$, we now assume that the two stochastic integrals $\int_0^\cdot \delta_s\cdot \ud \bar S_s:=\int_0^\cdot \delta_s\cdot \ud \left(\frac{S_s}{S_s^0}\right)$ and $\int_0^\cdot \delta_s\cdot \ud \hat S_s$ exist. Denote by $\bar C^\delta$ and $\hat C^\delta$ the cost processes associated to the strategy $\delta$ denominated in units of $S^0$ and $S^{\delta_*}$, respectively. 
\begin{lemma}\label{lem:c}
If $\bar C^\delta$ and $\hat C^\delta$ are the cost processes of the strategy $\delta$, then
\begin{equation} \label{eq:diffc}
\ud \hat C_t^\delta=\hat S_{t-}^0\ud \bar C_t^\delta+\ud [\bar C^\delta,\hat S^0]_t.
\end{equation}
\end{lemma} 
\begin{proof}
This result extends Lemma 3.1 in~\cite{bp}. 
For the reader's convenience we provide here briefly the proof of \eqref{eq:diffc}. It is formally analogous to the one of Lemma 3.1 in~\cite{bp}. By It\^o's formula, we have
\begin{align*}
&\ud \hat S_t^\delta\\
&=\ud \left(\frac{S_t^\delta}{S_t^{\delta_*}}\right) = \ud \left(\frac{S_t^\delta}{S_t^0}\cdot \frac{S_t^0}{S_t^{\delta_*}}\right)=\frac{S_{t-}^\delta}{S_{t-}^0}\ud \left(\frac{S_t^0}{S_t^{\delta_*}}\right) + \frac{S_{t-}^0}{S_{t-}^{\delta_*}} \ud \left(\frac{S_t^\delta}{S_t^0}\right) + \ud \left[\frac{S^\delta}{S^0},\frac{S^0}{S^{\delta_*}}\right]_t\\
& = \delta_{t-}\frac{S_{t-}}{S_{t-}^0}\ud \left(\frac{S_t^0}{S_t^{\delta_*}}\right) + \frac{S_{t-}^0}{S_{t-}^{\delta_*}} \ud \left(\frac{S_t^\delta}{S_t^0}\right) + \ud \left[\frac{S^\delta}{S^0},\frac{S^0}{S^{\delta_*}}\right]_t.
\end{align*}
Since 
$$
\ud \left(\frac{S_t^\delta}{S_t^0}\right) = \delta_{t-} \ud \left(\frac{S_t}{S_t^0}\right) + \ud \bar C_t^\delta,
$$
then
$$
\ud \left[\frac{S^\delta}{S^0},\frac{S^0}{S^{\delta_*}}\right]_t=\delta_{t-}\ud \left[\frac{S}{S^0},\frac{S^0}{S^{\delta_*}}\right]_t + \ud \left[\bar C^\delta,\frac{S^0}{S^{\delta_*}}\right]_t.
$$
Finally
\begin{align*}
\ud \hat S_t^\delta & = \delta_{t-} \left\{\frac{S_{t-}}{S_{t-}^0}\ud \left(\frac{S_t^0}{S_t^{\delta_*}}\right) + \frac{S_{t-}^0}{S_{t-}^{\delta_*}} \ud \left(\frac{S_t}{S_t^0}\right) + \ud \left[\frac{S}{S^0},\frac{S^0}{S^{\delta_*}}\right]_t\right\}\\
& \quad \quad + \frac{S_{t-}^0}{S_{t-}^{\delta_*}}\ud \bar C_t^\delta + \ud \left[\bar C^\delta,\frac{S^0}{S^{\delta_*}}\right]_t\\
&=\delta_{t-}\ud \left(\frac{S_t}{S_t^{\delta_*}}\right)+\frac{S_{t-}^0}{S_{t-}^{\delta_*}}\ud \bar C_t^\delta + \ud \left[\bar C^\delta,\frac{S^0}{S^{\delta_*}}\right]_t\\
&=\delta_{t-}\ud \hat S_t^\delta + \hat S_{t-}^0\ud \bar C_t^\delta + \ud \left[\bar C^\delta,\hat S^0\right]_t.
\end{align*}
\end{proof}
\noindent By Lemma \ref{lem:c}, the cost process of a risk-minimizing strategy (with respect to a given discounting factor) is given by a $\P$-martingale. This property provides a fundamental characterization of (local) risk-minimizing strategies (with respect to a given discounting factor). Here we show that they are invariant under a change of numéraire.
\begin{proposition} \label{prop:icn}
Under the same hypotheses of the previous lemma, if the
process $\bar C^\delta$ is a continuous $\P$-local martingale strongly orthogonal to the martingale part of $\bar S$, then $\hat C^\delta$ is also a (continuous) $\P$-local martingale strongly orthogonal to the martingale part of $\hat S$.
\end{proposition}

\begin{proof}
This result generalizes Proposition 3.1 of~\cite{bp}. The proof essentially follows from It\^o's formula and Lemma \ref{lem:c}.
From integration by parts formula, we have that
$$
\ud \left(\frac{S_t}{S_t^{\delta_*}}\right)=\ud \left(\frac{S_t}{S_t^0}\cdot \frac{S_t^0}{S_t^{\delta_*}}\right)=\frac{S_{t-}}{S_{t-}^0}\ud \left(\frac{S_t^0}{S_t^{\delta_*}}\right)+\frac{S_{t-}^0}{S_{t-}^{\delta_*}}\ud \left(\frac{S_t}{S_t^0}\right)+\ud \left[\frac{S}{S^0}, \frac{S^0}{S^{\delta_*}}\right]_t
$$
where by It\^o's formula
\begin{equation*}
\begin{split}
\ud \left(\frac{S_t^0}{S_t^{\delta_*}}\right)&=\ud \left[\left(\frac{S_t^{\delta_*}}{S_t^0}\right)^{-1}\right]=-\left(\frac{S_{t-}^0}{S_{t-}^{\delta_*}}\right)^2\ud \left(\frac{S_t^{\delta_*}}{S_t^0}\right)+\left(\frac{S_{t-}^0}{S_{t-}^{\delta_*}}\right)^3\ud \left[\frac{S^{\delta_*}}{S^0},\frac{S^{\delta_*}}{S^0}\right]_t\\
& \quad \quad + \underbrace{
\Delta \left(\frac{S_t^{\delta_*}}{S_t^0}\right)^{-1}+\left(\frac{S_{t-}^0}{S_{t-}^{\delta_*}}\right)^2\Delta \left(\frac{S_{t}^{\delta_*}}{S_{t}^0}\right)-\left(\frac{S_{t-}^0}{S_{t-}^{\delta_*}}\right)^3\left[\Delta \left(\frac{S_{t}^{\delta_*}}{S_{t}^0}\right)\right]^2
}_{:=\ud \Sigma_t}.
\end{split}
\end{equation*}
From Lemma \ref{lem:c}, we have
\begin{align}
&\ud \left[\hat C^\delta, \frac{S}{S^{\delta_*}}\right]_t\nonumber\\
&=\frac{S_{t-}^0}{S_{t-}^{\delta_*}}\ud \left[\bar C^\delta, \frac{S}{S^{\delta_*}}\right]_t+\underbrace{\ud \left[\left[\bar C^\delta,\frac{S^0}{S^{\delta_*}}\right],\frac{S}{S^{\delta_*}}\right]_t}_{0} \nonumber \\
&=\frac{S_{t-}}{S_{t-}^{\delta_*}}\ud \left[\bar C^\delta, \frac{S^0}{S^{\delta_*}}\right]_t+\underbrace{\left(\frac{S_{t-}^0}{S_{t-}^{\delta_*}}\right)^2\ud \left[\bar C^\delta, \frac{S}{S^0}\right]_t}_{0}+\underbrace{\frac{S_{t-}^0}{S_{t-}^{\delta_*}}\ud \left[\bar C^\delta, \left[\frac{S}{S^0},\frac{S^0}{S^{\delta_*}}\right]\right]_t}_{0} \nonumber\\
& \quad \quad + \underbrace{\ud \left[\left[\bar C^\delta,\frac{S^0}{S^{\delta_*}}\right],\frac{S}{S^{\delta_*}}\right]_t}_{0} \nonumber\\
&=\frac{S_{t-}}{S_{t-}^{\delta_*}}\ud \left[\bar C^\delta, \frac{S^0}{S^{\delta_*}}\right]_t,  \label{eq:brcd}
\end{align}
where we have used the fact that $\bar C^\delta$ is a continuous $\P$-local martingale strongly orthogonal to the martingale part of $\bar S$.
Furthermore,
\begin{equation}\label{eq:brcd1}
\begin{split}
&\ud \left[\bar C^\delta, \frac{S^0}{S^{\delta_*}}\right]_t\\
&=-\left(\frac{S_{t-}^0}{S_{t-}^{\delta_*}}\right)^2\ud \left[\bar C^\delta, \frac{S^{\delta_*}}{S^0}\right]_t+\left(\frac{S_{t-}^0}{S_{t-}^{\delta_*}}\right)^3 \ud\left[\bar C^\delta, \left[\frac{S^{\delta_*}}{S^0},\frac{S^{\delta_*}}{S^0}\right]\right]_t+ \ud \left[\bar C^\delta, \Sigma \right]_t \\
&=-\left(\frac{S_{t-}^0}{S_{t-}^{\delta_*}}\right)^2\ud \left[\bar C^\delta, \frac{S^{\delta_*}}{S^0}\right]_t.
\end{split}
\end{equation}
If now $\bar C^\delta$ is a $\P$-local martingale strongly orthogonal to $\bar S$, then 
$$
\ud \left [ \bar C^\delta, \frac{S^{\delta_*}}{S^0}\right]_t=0, \quad t \in [0,T].
$$
By \eqref{eq:brcd} and \eqref{eq:brcd1}, we have that also $\ud \left[\hat C^\delta, \frac{S}{S^{\delta_*}}\right]_t
=0$, hence $\hat C^\delta$ is strongly orthogonal to the martingale part of $S$. This concludes the proof.
\end{proof}

\noindent Now it is possible to state the main result that guarantees that invariance under change of numéraire is kept in the case of right-continuous asset price processes if we assume that the cost process $\bar C^\delta$ is continuous.

\begin{theorem}
Let $\delta$ be an $L^2$-admissible strategy with respect to the numéraires
$S^0$ and $S^{\delta_*}$ and assume that $\bar C^\delta$ is continuous. If $\delta$ is locally risk-minimizing under the numéraire $S^0$, then
$\delta$ is locally risk-minimizing also with respect to the numéraire  $S^{\delta_*}$, i.e. it is benchmarked locally risk-minimizing.
\end{theorem} 

\begin{proof}
The proof is an immediate consequence of Proposition \ref{prop:icn}: if $\delta$ is a locally risk-minimizing strategy under $S^0$, the cost $\hat C^\delta$ is a $\P$-local martingale strongly orthogonal to the martingale part of
$\hat S$. But since the
strategy is $L^2$-admissible with respect to $S^{\delta_*}$, the cost process $\hat C^\delta$ is actually a square integrable
$\P$-martingale.
\end{proof}


\subsection{Benchmarked local risk-minimization under incomplete information}

Here we show an example of benchmarked local risk-minimization that works under general assumptions on $\hat S$.
Similarly to~\cite{fs}, we consider a situation where the financial market would be complete if we had more information. The available information is described by the filtration $\mathbb F$. We suppose that the benchmarked claim $\hat H$ is attainable with respect to some larger filtration. Only at the terminal time $T$, but not at times $t < T$, all the information relevant for a perfect hedging of a claim will be available to us. So let $\tilde{\mathbb F}:=(\tilde \F_t)_{0 \leq t \leq T}$ be a right-continuous filtration such that 
$$
\F_t \subseteq \tilde \F_t \subseteq \F, \quad t \in [0,T].
$$
We now show how the results of~\cite{fs} hold without assuming that the underlying asset price processes are continuous, if we consider a general benchmarked market. Note furthermore that we are not going to assume that the benchmarked assets are $\P$-local martingales.\\
Consider now the benchmarked asset price process $\hat S$. Then $\hat S$ is a $\P$-supermartinga-\linebreak le and admits the Doob-Meyer's decomposition
$$
\hat S_t=M_t-A_t, \quad t \in [0,T],
$$
where $A$ is an $\bF$-predictable increasing finite variation process and $M$ is a $\P$-local martingale.
\begin{ass}
\noindent Suppose now that the vector process $\hat S$ belongs to $\mathcal S^2(\P)$, that is, the space of $\P$-semimartingales satisfying the integrability condition
\begin{equation*}\label{eq:integrcond}
\esp{\hat S_0^2+[M]_T+|A|_T^2}< \infty,
\end{equation*}
where $|A|=\{|A|_t:\ t \in [0,T]\}$ is the total variation of $A$.
In addition, the decomposition \eqref{supermartdeco} of $\hat S$ with respect to $\mathbb F$ is still valid with respect to $\tilde{\mathbb F}$. In other words we assume that $M$ is a $\P$-martingale with respect to $\tilde{\mathbb F}$, although it is adapted to the smaller filtration $\mathbb F$.
\end{ass}
\noindent Suppose now that $\hat H \in L^2(\F_T,\P)$ is attainable with respect to the larger filtration $\tilde{\mathbb F}$, i.e. 
\begin{equation} \label{eq:attainable}
\hat H=\tilde H_0+\int_0^T\tilde \xi_s^{\hat H}\ud \hat S_s,
\end{equation}
where $\tilde H_0$ is $\tilde \F_0$-measurable and the process $\tilde \xi^{\hat H}=\{\tilde \xi_t^{\hat H}=(\tilde \xi_t^{\hat H,0},\tilde \xi_t^{\hat H,1},\ldots,\tilde \xi_t^{\hat H,d})^\top,$ $t$ $\in [0,T]\}$ is predictable with respect to $\tilde{\mathbb F}$. We now need to specify suitable integrability conditions.
\begin{ass} \label{ass:intcond}
We suppose that the $(\tilde \bF,\P)$-semimartingale
$$
\tilde H_0+\int_0^t\tilde \xi_s^{\hat H}\ud \hat S_s, \quad t \in [0,T],
$$
associated to $\hat H$ belongs to the space $\mathcal S^2(\P)$, i.e.
\begin{equation}\label{eq:integrcond1}
\esp{\tilde H_0^2+\int_0^T(\tilde \xi_s^{\hat H})^\top\ud [M]_s\tilde \xi_s^{\hat H}+\left(\int_0^T|(\tilde \xi_s^{\hat H})^\top||\ud A_s|\right)^2}< \infty.
\end{equation}
\end{ass}
\begin{theorem} \label{th:incomplete}
Suppose that $\hat H$ satisfies \eqref{eq:attainable} and \eqref{eq:integrcond1}. Then $\hat H$ admits the representation 
\begin{equation} \label{eq:hatHdecomp}
\hat H=\tilde H_0+\int_0^T\xi_s^{\hat H}\ud \hat S_s+L_T^{\hat H},
\end{equation}
with $\hat H_0=E[\tilde H_0|\F_0]$, where
$$
\xi^{\hat H}:={}^p(\tilde \xi^{\hat H})
$$
is the $\mathbb F$-predictable projection of the $\tilde{\mathbb F}$-predictable vector process $\tilde \xi^{\hat H}$, and where $L^{\hat H}=\{L_t^{\hat H}=(L_t^{\hat H,0},L_t^{\hat H,1},\ldots,L_t^{\hat H,d})^\top,\ t \in [0,T]\}$ is the square-integrable $(\bF,\P)$-martingale, orthogonal to $M$, associated to
$$
L_T^{\hat H}:=\tilde H_0-\hat H_0+\int_0^T(\tilde \xi_s^{\hat H}-\xi_s^{\hat H})\ud \hat S_s \in L^2(\F_T,\P).
$$
\end{theorem}

\begin{proof}
\noindent {\bf Step 1.} First we need to check that all components in \eqref{eq:hatHdecomp} are square-integrable.  Denote by ${}^p X$ the (dual) $\bF$-predictable projection of a process $X$. By \eqref{eq:integrcond1} and by the properties of the $\mathbb F$-predictable projection (see~\cite{dm2}, VI.57),
\begin{equation*}
\infty > \esp{\left(\int_0^T\tilde \xi_s^{\hat H}\ud M_s\right)^2}=\esp{\int_0^T(\tilde \xi_s^{\hat H})^2\ud [M]_s}=\esp{\int_0^T{}^p\left((\tilde \xi_s^{\hat H})^2\right)\ud \langle M\rangle_s},
\end{equation*}
where the last equality holds since $\langle M\rangle$ is the $\bF$-predictable dual projection (see e.g.~\cite{dm2}, VI.73) of $[M]$. By Jensen's inequality we have
$$
\left({}^p(\tilde \xi_\tau^{\hat H})\right)^2=(\xi_\tau^{\hat H})^2=\left(E[\tilde \xi_\tau^{\hat H}|\F_{\tau-}]\right)^2 \leq E[(\tilde \xi_\tau^{\hat H})^2|\F_{\tau-}]={}^p\left((\tilde \xi_\tau^{\hat H})^2\right),
$$
for every predictable $\bF$-stopping time $\tau$, on the set $\{\tau < \infty\}$. Hence if we consider $\tau=s$, again by the properties of the $\mathbb F$-predictable dual projection, it follows that
$$
\esp{\int_0^T{}^p\left((\tilde \xi_s^{\hat H})^2\right)\ud \langle M\rangle_s} \ge \esp{\int_0^T\left({}^p(\tilde \xi_s^{\hat H})\right)^2\ud \langle M\rangle_s} =  \esp{\int_0^T(\xi_s^{\hat H})^2\ud \langle M\rangle_s}
$$
and finally
$$
\int_0^T \xi_s^{\hat H}\ud M_s \in L^2(\F_T,\P).
$$
In order to show that
$$
\int_0^T \xi_s^{\hat H}\ud A_s \in L^2(\F_T,\P),
$$
we prove that
$$
\left|\esp{Z_T\int_0^T\xi_s^{\hat H} \ud A_s}\right| \leq c \cdot \|Z_T\|_2,\quad c \in \R
$$
for any bounded $\F_T$-measurable random variable $Z_T$ with $L^2$-norm $\|Z_T\|_2$. Let $Z=\{Z_t,\ t \in [0,T]\}$ denote a right-continuous version with left-limits of the $(\bF,\P)$-martingale $E[Z_T|\F_t]$, $t\in [0,T]$, and put $Z^*=\sup_{0 \leq t \leq T}|Z_t|$. Since $A$ is $\bF$-predictable, we can use some properties of the $\mathbb F$-predictable projection (see VI.45 and VI.57 in~\cite{dm2}), and obtain
\begin{align*}
\left|\esp{Z_T \int_0^T \xi_s^{\hat H} \ud A_s}\right| & = \left|\esp{\int_0^T Z_{s-}\xi_s^{\hat H} \ud A_s}\right|\\
& = \left|\esp{\int_0^T Z_{s-}\tilde \xi_s^{\hat H} \ud A_s}\right|\\
& \leq \left|Z^*\esp{\int_0^T\tilde \xi_s^{\hat H} \ud |A|_s}\right|\\
& \leq \|Z^*\|_2\left\|\int_0^T\tilde \xi_s^{\hat H} \ud |A|_s\right\|_2\\
& \leq c\cdot \|Z^*\|_2, \quad c \in \R, 
\end{align*}
where in the last inequality we have used \eqref{eq:integrcond1} and Doob's inequality for the supremum of a square-integrable $\P$-martingale.\\
\noindent {\bf Step 2.} Clearly, $\tilde H_0-\hat H_0 \in L^2(\Omega,\tilde \F_0,\P)$ is orthogonal to all square-integrable stochastic integrals of $M$ with respect to the filtration $\tilde{\mathbb F}$, hence in particular with respect to the filtration $\mathbb F$. Then, it only remains to show that
\begin{equation}\label{eq:orthog}
\esp{\left(\int_0^T(\tilde \xi_s^{\hat H}-\xi_s^{\hat H})\ud \hat S_s\right)\left(\int_0^T\mu_s\ud M_s\right)}=0
\end{equation}
for all bounded $\mathbb F$-predictable processes $\mu=\{\mu_t,\ t \in [0,T]\}$. This will imply that the $(\bF,\P)$-martingale $L^{\hat H}$ is orthogonal to $M$. First we note that \eqref{eq:orthog} is equivalent to the following
\begin{equation}\label{eq:orthog1}
\esp{\left(\int_0^T\tilde \xi_s^{\hat H}\ud \hat S_s\right)\left(\int_0^T\mu_s\ud M_s\right)}=\esp{\left(\int_0^T\xi_s^{\hat H}\ud \hat S_s\right)\left(\int_0^T\mu_s\ud M_s\right)}.
\end{equation} 
Then we decompose the left-side of \eqref{eq:orthog1} into
$$
\esp{\left(\int_0^T\tilde \xi_s^{\hat H}\ud M_s\right)\left(\int_0^T\mu_s\ud M_s\right)}+\esp{\left(\int_0^T\tilde \xi_s^{\hat H}\ud A_s\right)\left(\int_0^T\mu_s\ud M_s\right)}.
$$
We have
\begin{equation} \label{eq:add1}
\esp{\left(\int_0^T\tilde \xi_s^{\hat H}\ud M_s\right)\left(\int_0^T\mu_s\ud M_s\right)}=\esp{\int_0^T\tilde \xi_s^{\hat H}\cdot \mu_s \ud [M]_s}
\end{equation}
and
\begin{equation} \label{eq:add2}
\esp{\left(\int_0^T\tilde \xi_s^{\hat H}\ud A_s\right)\left(\int_0^T\mu_s\ud M_s\right)}=\esp{\int_0^T\tilde \xi_s^{\hat H}\left(\int_{[0,s)}\mu_u\ud M_u\right)\ud A_s},
\end{equation}
by the property of the $\bF$-predictable projection (see VI.45 in~\cite{dm2}).
Since $\langle M\rangle$ is the $\bF$-predictable dual projection of $[M]$, we can rewrite \eqref{eq:add1} as follows
\begin{align} \label{eq:prpr}
\esp{\int_0^T\tilde \xi_s^{\hat H}\cdot \mu_s \ud [M]_s}&=\esp{\int_0^T{}^p(\tilde \xi_s^{\hat H}\cdot \mu_s) \ud \langle M\rangle_s} \nonumber \\
&=\esp{\int_0^T{}^p(\tilde \xi_s^{\hat H})\cdot \mu_s \ud \langle M\rangle_s}\\
&=\esp{\int_0^T\xi_s^{\hat H}\cdot \mu_s \ud \langle M\rangle_s},\nonumber
\end{align}
where the second equality \eqref{eq:prpr} follows from Remark 44(e) in~\cite{dm2}. Now, from the properties of the $\bF$-predictable projection it is clear that $\tilde \xi^{\hat H}$ can be replaced by $\xi^{\hat H}$ in \eqref{eq:add2}, and this yields \eqref{eq:orthog}.
\end{proof}

\section{Applications} \label{continuousmarket}

In the remaining part of the paper we discuss some examples that illustrate how
classical (local) risk-minimization is generalized to benchmarked (local) risk-minimization in a market when there is no equivalent risk-neutral probability measure. Finally, we will demonstrate that the presence of jumps
does not create a major problem, which is not easily resolved under classical (local) risk-minimization.

\subsection{Benchmarked Risk-Minimization in the Minimal Market Model with Random Scaling} \label{defaultablebond}

\noindent The notion of a minimal market model (in short MMM) is due to E. Platen and has been introduced in a series of papers with various co-authors; see Chapter 13 of~\cite{ph} for a recent textbook account. 
The version of the MMM described here, which generalizes the stylized version
derived in Section 13.2 of~\cite{ph}, is governed by a particular choice of the discounted numéraire portfolio drift.
The MMM generates stochastic volatilities that involve transformations of squared Bessel processes. \\
We begin by describing a continuous financial market model almost similarly as in Chapter 10 of~\cite{ph}. More precisely, in this framework uncertainty is modeled by $d$ independent standard Wiener processes $W^k=\{W_t^k,\ t \in[0,T]\}$, $k \in \{1,2,\ldots,d\}$ on $(\Omega,\F,\P,\bF)$, where
$\bF:=(\F_t)_{0\leq t \leq T}$, with $\F_t=\F_t^{W^1}\vee\F_t^{W^2}\vee \ldots \vee \F_t^{W^d}$ for each $t \in [0,T]$.
We assume that the value at time $t$ of the savings account $S^0$ is given by
\begin{equation*} \label{eq:savings}
S_t^0=\exp{\left\{\int_0^t r_s\ud s\right\}} < \infty
\end{equation*}
for $t\in[0,T]$, where $r=\{r_t,\ t\in [0,T]\}$ denotes the $\bF$-adapted short term interest rate. For simplicity, we suppose $r_t=r \ge 0$, for every $t \in [0,T]$. Furthermore, we assume 
that the dynamics of the primary security account processes $S^j=\{S_t^j,\ t \in [0,T]\}$, $j\in {1,2,\ldots,d}$, are given by the SDE
\begin{equation} \label{eq:dyn}
\ud S_t^j=S_t^j\left(a_t^j\ud t + \sum_{k=1}^d b_t^{j,k}\ud W_t^k\right)
\end{equation}
for $t \in [0,T]$ with $S_0^j >0$. The $j$-th appreciation rate $a^j=\{a_t^j,\ t\in[0,T]\}$ and the $(j,k)$-th volatility $b^{j,k}=\{b_t^{j,k},\ t\in [0,T]\}$
are $\bF$-predictable processes for $j,k \in \{1,2,\ldots,d\}$ satisfying suitable integrability conditions. Furthermore the volatility matrix ${\bf b}_t=[b_t^{j,k}]_{j,k=1}^{d}$ is for Lebesgue almost-every $t\in[0,T]$ assumed to be invertible. This assumption avoids redundant primary security accounts and also ensures the existence of the numéraire portfolio. By introducing the appreciation rate vector $a_t=(a_t^1,a_t^2,\ldots,a_t^d)^\top$ and the unit vector ${\bf 1}=(1,1,\ldots,1)^\top$, we obtain the {\em market price of risk} vector
\begin{align}\label{eq:mpr}
\theta_t&=(\theta_t^1,\theta_t^2,\ldots,\theta_t^d)={\bf b}_t^{-1}[a_t-r_t{\bf 1}]
\end{align}
for $t\in[0,T]$. The notion \eqref{eq:mpr} allows us to rewrite the SDE \eqref{eq:dyn} in the form
\begin{equation}\label{eq:dyn1}
\ud S_t^j=S_t^j\left\{r_t\ud t + \sum_{k=1}^d(\theta_t^k -  \sigma_t^{j,k})[\theta_t^k \ud t + \ud W_t^k]\right\}
\end{equation}
with $(j,k)$-th volatility
$$
\sigma_t^{j,k}=\theta_t^k-b_t^{j,k}
$$
for $t\in[0,T]$ and $j,k \in \{1,2,\ldots, d\}$. 
According to~\cite{ph}, Chapter 10, it is easy to check that the numéraire portfolio satisfies the SDE
\begin{equation} \label{eq:numéraire portfolio}
\ud S_t^{\delta_*}=S_t^{\delta_*}\left[r_t\ud t + \sum_{k=1}^d \theta_t^k\left(\theta_t^k\ud t + \ud W_t^k\right)\right]
\end{equation}
for $t \in [0,T]$, where we set $S_0^{\delta_*}=1$. By \eqref{eq:numéraire portfolio}, it follows that  
the risk premium of the numéraire portfolio equals the square of its volatility. Denote by $\bar S^{\delta_*}$ the discounted numéraire portfolio, i.e. 
$$
\bar S_t^{\delta_*}=\frac{S_t^{\delta_*}}{S_t^0}, \quad t \in [0,T].
$$
It is easy to check that $\bar S^{\delta_*}$ satisfies the SDE
\begin{equation} \label{eq:sde-np}
\ud \bar S_t^{\delta_*}=\ud (\hat S_t^0)^{-1}=(\hat S_t^0)^{-1} |\theta_t|(|\theta_t| \ud t+ \ud W_t)=\bar S_t^{\delta_*}|\theta_t|(|\theta_t| \ud t+ \ud W_t),
\end{equation}
where
$$
\ud W_t=\frac{1}{|\theta_t|} \sum_{k=1}^d
\theta_t^k \ud W_t^k
$$
is the stochastic differential of a standard Wiener process $W$ and $|\theta|$ denotes the total market price of risk. For the efficient
modeling of the numéraire portfolio it is important to find an appropriate parametrization. Let us parametrize the discounted
numéraire portfolio dynamics, that is the SDE \eqref{eq:sde-np}, by its trend. More precisely, we
consider the discounted numéraire portfolio drift
\begin{equation}\label{eq:drift}
\alpha_t=\bar S_t^{\delta_*}|\theta_t|^2
\end{equation}
for $t \in [0,T]$. Using this parametrization obtained from \eqref{eq:drift}, we can rewrite the SDE \eqref{eq:sde-np} of the discounted numéraire portfolio as follows:
\begin{equation}\label{eq:SDEparam}
\ud \bar S_t^{\delta_*}=\alpha_t\ud t+ \sqrt{\bar S_t^{\delta_*}\alpha_t}\ud W_t, \quad t \in [0,T].
\end{equation}
According to Section 13.4 of~\cite{ph}, we now assume that 
the discounted numéraire portfolio drift
is given by
\begin{equation}\label{eq:numéraire portfoliodrift}
\alpha_t=\left(\frac{\delta}{2}-1\right)^2 \gamma_t Z_t^{\frac{\delta - 4}{2}}, \quad t \in [0,T],
\end{equation}
where $Z$ is a a squared
Bessel process with a general dimension $\delta > 2$ satisfying the SDE
\begin{equation}\label{eq:z}
\ud Z_t = \frac{\delta}{4} \gamma_t \ud t + \sqrt{\gamma_t Z_t}\ud W_t
\end{equation}
for $t \in [0,T]$ with $Z_0 > 0$. Note that with this choice of $\alpha$, we have that 
\begin{equation} \label{def:scaling}
Z_t=\left(\bar S_t^{\delta_*}\right)^{\frac{2}{\delta - 2}},
\end{equation}
for $t \in [0,T]$ and $\delta \in (2,\infty)$. We should stress that for the standard choice $\delta = 4$ and $\gamma_t = 1$ for every $t \in [0,T]$, we recover the
stylized MMM, see Sect.13.2 of~\cite{ph}. 
Note also that for the standard case with
$\delta = 4$ the discounted numéraire portfolio drift does not depend on $Z_t$. According to~\cite{ph}, we assume here that the scaling process $\gamma$ is a nonnegative, $\bF$-adapted stochastic process that satisfies a
SDE of the form 
\begin{equation}\label{eq:scalingGAMMA}
\ud \gamma_t =a(t,\gamma_t)\ud t + b(t,\gamma_t)\left(\rho \ud W_t + \sqrt{1-\rho^2}\ud \tilde W_t \right), \quad t \in [0,T],
\end{equation}
with a random initial value $\gamma_0 > 0$. Here $\tilde W$ is a Wiener process
that models some uncertainty in trading activity and is assumed to be independent
of $W$. The scaling correlation $\rho$ is,
for simplicity, assumed to be constant. Under this formulation the dynamics of the diffusion process $\gamma$ can be chosen to match empirical evidence, see Section 13.4 of~\cite{ph} for some examples. 
Note that an equivalent risk neutral probability
measure does not exist for the above model. The benchmarked savings
account $\hat S^0$ and, thus, the candidate Radon-Nikodym derivative process
$\Lambda=\{\Lambda_t,\ t \in [0,T]\}$ with
$$
\Lambda_t=\left(\frac{\bar S_t^{\delta_*}}{\bar S_0^{\delta_*}}\right)^{-1}=\left(\frac{Z_t}{Z_0}\right)^{1-\frac{\delta}{2}}
$$
are by (8.7.24) of~\cite{ph} strict $\bF$-local martingales when we assume no correlation, that
is $\rho = 0$. Note also that $\hat S^0$ is not square-integrable, see formula (8.7.14) in~\cite{ph}. \\
We now consider an application of risk-minimization in the MMM with random scaling.
For the sake of simplicity, we refer to a generalized MMM where we can find two primary security account processes $S^1$ and $S^2$, whose dynamics are described at each time $t \in [0,T]$ by the SDE \eqref{eq:dyn}, for $j=1,2$.
We recall that the primary security account processes $S^j$, $j=1,2$, are driven by the two independent standard Wiener processes $W^1$ and $W^2$, i.e. their behavior is described by the following SDE:  
\begin{equation} \label{eq:dyn2}
\ud S_t^j=S_t^j\left(a_t^j\ud t + b_t^{j,1}\ud W_t^1 + b_t^{j,2}\ud W_t^2\right), \quad t \in [0,T],
\end{equation}
with $S_0^j >0$. By applying It\^o's formula, 
it is easy to compute the following SDE:
\begin{equation*} \label{eq:hatSj}
\ud \hat S_t^j = -\hat S_t^j|\theta_t|\ud W_t + \hat S_t^j\left(b_t^{j,1}\ud W_t^1 + b_t^{j,2}\ud W_t^2\right), \quad t \in [0,T], \ j \in \{1,2\}.
\end{equation*}
Then, $\hat S^j$ is an $\bF$-(local) martingale for both $j=1,2$. We assume that $\esp{\langle \hat S^j\rangle_T}< \infty$, for $j\in\{1,2\}$. Then by Corollary 4 on page 74 of Chapter II of~\cite{pp}, we obtain that $\esp{(\hat S_T^j)^2}=\esp{[\hat S^j]_T} < \infty$ and that $\hat S^j$ is a square-integrable $\bF$-martingale for both $j=1,2$.\\
\noindent If we now restrict our attention to the market given by the numéraire portfolio $S^{\delta_*}$ and the savings account $S^0$, the primary security account processes $S^1$ and $S^2$ cannot be perfectly replicated by investing in a portfolio containing only these two tradeable assets. 
Consider $\hat H=\hat S_T^j=\frac{S_T^j}{S_T^0 Z_T^{\frac{\delta}{2} - 1}}\in L^2(\F_T,\P)$ and let $\varphi=(\xi, \eta)=\{(\xi_t,\eta_t),\ t \in [0,T]\}$ be an $L^2$-admissible strategy such that
$$
S_T^j = \xi_T S_T^{\delta_*} + \eta_Te^{rT} = \xi_TS_T^0Z_T^{\frac{\delta}{2} - 1} + \eta_Te^{rT}, \quad j \in \{1,2\}.
$$
Then
$$
\hat H = \hat S_T^j = \xi_T+ \eta_T e^{rT}\frac{1}{S_T^0 Z_T^{\frac{\delta}{2} - 1}}=\xi_T + \eta_T \hat S_T^0, \quad j \in \{1,2\}.
$$ 
At time $t < T$, we have that
$$
\hat S_t^j = \xi_t + \eta_t\hat P_T(t,Z_t,\gamma_t),
$$
where 
\begin{equation} \label{eq:PT}
\hat P_T(t,Z_t,\gamma_t):=\condespfw{\frac{1}{S_T^{\delta_*}}}=\condespfw{\frac{1}{S_T^0 Z_T^{\frac{\delta}{2} - 1}}}
\end{equation} 
is the benchmarked fair price at time $t$ of a zero coupon bond with maturity $T$ and $\bF^W=(\F_t^W)_{0\leq t \leq T}$ denotes the natural filtration of $W$. We replicate $S^j$ by using $S^{\delta_*}$ and the money market account, or equivalently, we replicate the benchmarked primary security account $\hat S^j$ by investing in the benchmarked zero coupon bond and in $1$. In general, we do not have an explicit joint density of $(Z_T,\gamma_T)$, which we
would need to calculate the conditional expectation in \eqref{eq:PT}. However, it is possible to characterize the benchmarked fair zero coupon bond
pricing function $\hat P_T(\cdot, \cdot, \cdot)$ as the solution of a Kolmogorov backward equation and provide its description by using numerical methods for solving partial differential
equations (PDEs), see Section 15.7 in~\cite{ph}. Since $\hat P_T(\cdot, \cdot, \cdot)$ is of the form \eqref{eq:PT}, then it will admit a Brownian martingale representation, i.e.
\begin{equation*}\label{eq:martREPRE}
\hat P_T(t,Z_t,\gamma_t) = c + \int_0^t \psi_s \ud W_s, \quad t \in [0,T]
\end{equation*}
for a sufficiently integrable $\bF^W$-predictable process $\psi=\{\psi_t,\ t \in [0,T]\}$ and $c >0$. Without loss of generality, we can here assume that
\begin{equation}\label{eq:phit}
\forall t \in [0,T],\ \psi_t \neq 0 \quad \P-\mbox{a.s.}.
\end{equation}
Note that  $\hat P_T(\cdot, \cdot, \cdot)$ is also an $\bF$-martingale, i.e. for $t \in [0,T]$
\begin{align*}
\hat P_T(t,Z_t,\gamma_t) & = \condespfw{\frac{1}{S_T^0 Z_T^{\frac{\delta}{2} - 1}}} = \condespfww{\frac{1}{S_T^0 Z_T^{\frac{\delta}{2} - 1}}}\\
& = c + \int_0^t \psi_s \frac{\theta_s^1 \ud W_s^1 + \theta_s^2 \ud W_s^2}{|\theta_s|},
\end{align*}
or equivalently
\begin{equation}\label{eq:diffP}
\ud \hat P_T(t,Z_t,\gamma_t) =\psi_t \ud W_t = \frac{\psi_t \theta_t^1}{|\theta_t|}\ud W_t^1 + \frac{\psi_t \theta_t^2}{|\theta_t|}\ud W_t^2.
\end{equation}
Our aim is then to perform benchmarked risk-minimization of the benchmarked security $\hat S^j$ with respect to $\hat P_T(\cdot, \cdot, \cdot)$. Note that
$$
W_t=\int_0^t\frac{\theta_s^1}{|\theta_s|}\ud W_s^1 + \int_0^t\frac{\theta_s^2}{|\theta_s|}\ud W_s^2, \quad t \in [0,T],
$$
is an $\bF$-(local) martingale.
Since $\hat P_T(\cdot, \cdot, \cdot)$ is continuous and $\hat S_T^j$ is square-integrable, the F\"ollmer-Schweizer decomposition of $\hat S_T^j$
with respect to $\hat P_T(\cdot, \cdot, \cdot)$ is given by 
the GKW decomposition: 
\begin{equation}\label{eq:GKWsj}
\hat S_T^j = \hat S_0^j + \int_0^T \eta_s\ud \hat P_T(s,Z_s,\gamma_s) + M_T, \quad \P-{\rm a.s.}, \quad 
\ j \in \{1,2\},
\end{equation}
where $M$ is a square-integrable $\bF$-martingale strongly orthogonal to $\hat P_T(\cdot, \cdot, \cdot)$ by Theorem \ref{prop:fs} and $\eta$ is an $L^2$-admissible strategy. Since it is easy to check that 
$$
W_t^\perp:=\int_0^t\frac{\theta_s^2}{|\theta_s|}\ud W_s^1 - \int_0^t\frac{\theta_s^1}{|\theta_s|}\ud W_s^2, \quad t \in [0,T],
$$
is strongly orthogonal to $W$, we may assume that $M$ is of the form
$$
M_t=\int_0^t\nu_s\ud W_s^\perp, \quad t \in [0,T],
$$
for a suitable process $\nu$. In particular by \eqref{eq:GKWsj} we also obtain that
\begin{equation} \label{eq:GKWsj2}
\hat S_t^j = \condespf{\hat S_T^j}=\hat S_0^j + \int_0^t \eta_s\ud \hat P_T(s,Z_s,\gamma_s) + \int_0^t\nu_s\ud W_s^\perp, \quad \P-{\rm a.s.}, \quad t \in [0,T],
\end{equation}
since $\hat S^j$ is a square-integrable $\bF$-martingale for both $j=1,2$.
We now identify $(\eta,\nu)$ 
and the associated cost process.
By decomposition \eqref{eq:GKWsj2} and representation \eqref{eq:diffP}, for each $t \in [0,T]$, we get
$$
\ud \left[\hat S^j,\hat P_T(\cdot, \cdot, \cdot)\right]_t=\eta_t\ud [\hat P_T(\cdot, \cdot, \cdot)]_t= \eta_t\psi_t^2\ud t.
$$
On the other hand, taking \eqref{eq:hatSj} into account, for every $t \in [0,T]$, we have
$$
\ud \left[\hat S^j,\hat P_T(\cdot, \cdot, \cdot)\right]_t=\hat S_t^j\psi_t\left(\frac{\theta_t^1b_t^{j,1}}{|\theta_t|}+\frac{\theta_t^2b_t^{j,2}}{|\theta_t|}-|\theta_t|\right)\ud t.
$$
Then, by comparing the two relationships we define
\begin{equation} \label{eq:opt}
\eta_t=
\frac{\hat S_t^j}{\psi_t}\left(\frac{\theta_t^1b_t^{j,1}}{|\theta_t|}+\frac{\theta_t^2b_t^{j,2}}{|\theta_t|}-|\theta_t|\right)
\end{equation}
for every $t \in [0,T]$ since \eqref{eq:phit} holds. Note that the component $\eta$ is well-defined by \eqref{eq:GKWsj2}.
Analogously, if we compute the bracket process of $\hat S^j$ and $W^\perp$, for every $t \in [0,T]$ we get:
$$
\ud \left[\hat S^j,W^\perp\right]_t=\nu_t\ud[W^\perp]_t=\nu_t \ud t,
$$
and 
$$
\ud \left[\hat S^j,W^\perp\right]_t=\frac{\hat S_t^j}{|\theta_t|}\left(\theta_t^2b_t^{j,1}-\theta_t^1b_t^{j,2}\right) \ud t,
$$
from which we deduce that
\begin{equation*} 
\nu_t=\frac{\hat S_t^j}{|\theta_t|}\left(\theta_t^2b_t^{j,1}-\theta_t^1b_t^{j,2}\right), \quad t \in [0,T].
\end{equation*}
Hence, the  F\"ollmer-Schweizer decomposition for $\hat H$ is given by
$$
\hat H = \hat S_0^j + \int_0^T \eta_s\ud \hat P_T(s,Z_s,\gamma_s) + \int_0^T\nu_s\ud W_s^\perp, \quad \P-{\rm a.s.},
$$
where $\eta$ defined in \eqref{eq:opt} is the benchmarked risk-minimizing strategy and 
$$
\hat C_t= \hat S_0^j + \int_0^t\frac{\hat S_s^j}{|\theta_s|}\left(\theta_s^2b_s^{j,1}-\theta_s^1b_s^{j,2}\right)\ud W_s^\perp, \quad t \in [0,T],
$$
is the optimal benchmarked cost process.

\subsection{Benchmarked Risk-Minimization in the Stylized Minimal Market Model}

\noindent The benchmarked risk-minimizing strategy $\eta$ given in \eqref{eq:opt} strictly depends on the process $\psi$, see \eqref{eq:diffP},
that cannot be computed explicitly in the general minimal market model described above. However, if we perform the benchmarked risk-minimization for the benchmarked security $\hat S^j$ with respect to $\hat P(\cdot, \cdot, \cdot)$ under the stylized version of the MMM, i.e. where the discounted numéraire portfolio drift is an exponentially growing
function of time, it is possible to provide an explicit representation for the optimal strategy. If we take $\alpha$
to be a deterministic exponential function of time of the form 
$$
\alpha_t=\alpha_0\exp\{\beta t\}, \quad t \in [0,T],
$$
where $\alpha_0>0$ is a scaling parameter and $\beta>0$ denotes the long term net growth rate of the market, then
the stylized MMM corresponds indeed to choose in \eqref{def:scaling} $\delta=4$ and $\gamma_t=1$, for every $t \in [0,T]$, as observed previously. Suppose that $\esp{(\hat S_T^j)^2}< \infty$, for $j\in\{1,2\}$. Under the stylized version of the MMM where no dependence exists between the processes $Z$ and $\gamma$, the problem can be completely solved since the benchmarked price $\hat P(t,T)$ at time $t$ of a fair zero coupon bond with maturity $T$ 
is given by the explicit formula
\begin{equation} \label{bondprice}
\hat P(t,T)=\left(1-\exp{\left\{-(\hat S_t^0)^{-1}f(t)\right\}}\right)\hat S_t^0,
\end{equation}
where $f(t)=\frac{2\beta}{\alpha_0(\exp{\{\beta T\}}-\exp{\{\beta t\}})}$.
By applying It\^o's formula to \eqref{bondprice}, we obtain
\begin{equation*} \label{eq:bondDiff}
\ud \hat P(t,T)=-\hat P(t,T)\left(\hat S_t^0-f(t)\frac{e^{-\frac{f(t)}{\hat S_t^0}}}{1-e^{- \frac{f(t)}{\hat S_t^0}}}\right)\sqrt{(\hat S_t^0)^{-1}\alpha_t}\ud W_t, \quad t \in [0,T].
\end{equation*}
Proceeding as in the previous example and setting 
$$
\psi_t = -\hat P(t,T)\left(\hat S_t^0-f(t)\frac{e^{-\frac{f(t)}{\hat S_t^0}}}{1-e^{- \frac{f(t)}{\hat S_t^0}}}\right)\sqrt{(\hat S_t^0)^{-1}\alpha_t},
$$
for each $t \in [0,T]$, we obtain that the benchmarked risk-minimizing strategy is explicitly given by
$$
\eta_t=\frac{\hat S_t^j}{\alpha_t\hat P(t,T)}\left(|\theta_t|^2-\theta_t^1b_t^{j,1}-\theta_t^2b_t^{j,2}\right)\left(\hat S_t^0-f(t)\frac{e^{-\frac{f(t)}{\hat S_t^0}}}{1-e^{- \frac{f(t)}{\hat S_t^0}}}\right)^{-1}, \quad t \in [0,T].
$$

\subsection{Benchmarked Risk-Minimization for a Defaultable Put on an Index under the Stylized Minimal Market Model}

The numéraire portfolio $S^{\delta_*}$ can be realistically interpreted as a diversified equity index, see~\cite{ph}. Index linked variable annuities or puts on the numéraire portfolio are products that are of particular interest to pension plans. The recent financial crisis made rather clear that the event of a potential default of the issuing bank has to be taken into account. Set as before $\bF^W=(\F_t^W)_{0 \leq t \leq T}$. We now study the problem of pricing and hedging a defaultable put on the numéraire portfolio with strike $K \in \R_+$ and maturity $T \in (0,\infty)$ in the stylized MMM. Since the benchmarked payoff $\hat H$ of a put option is of the form $F(\hat S_T^0)$, for a bounded function $F$, then $\hat H \in L^2(\F_T^W,\P)$ even if $\hat S_t^0$ is not square-integrable for any $t \in [0,T]$. 
The fair default free benchmarked price $\hat p_{T,K}(t)$ at time $t$ is given by~\cite{ph} in the form 
\begin{align*}
&\hat p_{T,K}(t)\\
&=\condespfw{\frac{(K-S_T^{\delta_*})^+}{S_T^{\delta_*}}}=\condespfw{\left(K\hat S_T^0-1\right)^+}\\
&=\condespfw{\left(K\hat P(T,T)-1\right)^+}\\
&=-Z^2(d_1;4,l_2)-K\hat P(t,T)\left(Z^2(d_1;0,l_2)-\exp{\left\{-\frac{l_2}{2}\right\}}\right),
\end{align*}
where $Z^2(x;\nu,l)$ denotes the non-central chi-square distribution function with $\nu$ degrees of freedom, non-centrality parameter $l$ and which is taken at the level $x$. Here we have
$$
d_1=\frac{4 \eta K}{\alpha_t\left(\exp{\{\eta (T-t)\}}-1\right)}
$$
and
$$
l_2=\frac{2 \eta}{\alpha_t\left(\exp{\{\eta (T-t)\}}-1\right)\hat P(t,T)}.
$$
\noindent Now, we extend the stylized MMM 
to include default risk. 
Beyond the traded uncertainty given by the standard $\bF^W$-Wiener process $W$, there is also an additional source of randomness due to the presence of a possible default that, according to intensity based modeling, shall be modeled via a compensated jump process. More precisely, we assume that the random time of default $\tau$ is represented by a stopping time in the given filtration $\bF$. Let $D$ be the default process, defined as $D_t=\I_{\{\tau \leq t\}}$, for $t \in [0,T]$. We assume that $\tau$ admits an $\bF^W$-intensity, that is, there exists an $\bF^W$-adapted, nonnegative, (integrable) process $\lambda$ such that the process
\begin{equation*} \label{compens}
Q_t=D_t-\int_0^{\tau \wedge t} \lambda_s \ud s=D_t-\int_0^t \tilde\lambda_s \ud s, \quad t \in[0,T]
\end{equation*}
is a $\P$-martingale. Notice that
for the sake of brevity we have written $\tilde \lambda_t=\lambda_t\I_{\{\tau \ge t\}}$. In particular, we obtain that the existence of the intensity implies that $\tau$ is a totally inaccessible $\bF^W$-stopping time, see~\cite{dm2}, so that $\P(\tau = \tilde \tau) = 0$ for any $\bF^W$-predictable stopping time $\tilde \tau$. Furthermore, we suppose that the default time $\tau$ and the underlying Wiener process $W$, are independent.
When $\lambda$ is constant, $\tau$ is the moment of the first jump of a Poisson process. \\
Then, the benchmarked payoff of the defaultable put can be represented as follows: 
$$
\hat H=\left(K\hat S_T^0-1\right)^+\cdot \left(1+(\bar \delta-1)D_T\right),
$$
where $\bar \delta$ is supposed to be the random recovery rate. In particular, we assume that $\bar \delta$ is a random variable in $L^2(\F_T^W,\P)$ depending only on $T$ and $\tau$, i.e.
\begin{equation} \label{delta}
\bar \delta=h(\tau \wedge T),
\end{equation}
for some Borel function $h:(\R, \mathcal B(\R)) \rightarrow (\R, \mathcal B(\R))$, $0 \leq h\leq 1$. 
Here we focus on the case when an agent recovers a random part of the promised claim at maturity. Moreover, we obtain that $\hat H \in L^2(\F_T^W,\P)$. Thus, we can apply the results of Section \ref{lrm} to compute the decomposition \eqref{fsdecomp} for $\hat H$, i.e. 
the GKW decomposition of $\hat H$ with respect to 
the $\P$-local martingale $\hat P(\cdot,T)$. By applying the real world pricing formula, we obtain
the relationship
\begin{equation*}
\hat U_H(t)=\condespfw{\left(K\hat S_T^0-1\right)^+}\cdot\condespfw{1+(h(\tau \wedge T)-1)D_T}=\hat p_{T,K}(t)\cdot\Psi_t,
\end{equation*}
for all $t \in [0,T]$.  Now it only remains to compute $\Psi_t$. First we note that for each $t \in [0,T]$, we have 
\begin{align*}
\Psi_t
&=1+\condespfw{h(\tau \wedge T) D_T}-\condespfw{D_T}\\
&=1+\condespfw{h(\tau \wedge T) D_T}-\left(1-(1-F_T^\tau)\right)\tilde Q_t\\
&=\condespfw{h(\tau \wedge T) D_T}+(1-F_T^\tau)\tilde Q_t,
\end{align*}
with
$$
\tilde Q_t=\frac{1-D_t}{1-F_t^\tau}, \quad t \in [0,T],
$$
where $F^\tau$ stands for the cumulative distribution function of $\tau$. We assume that $F_t^\tau < 1$, for
every $t \in [0,T]$, so that $\tilde Q$ is well-defined.
We note that the second equality in the above derivation follows from Corollary 4.1.2 of~\cite{br}. By using the same arguments as in~\cite{bc1}, we obtain for every $t \in [0,T]$ the equation
\begin{equation} \label{psirepre}
\Psi_t=\esp{g(\tau)}+\int_{]0,t]}\left(\tilde g(s)-\frac{1-F_T^\tau}{1-F_s^\tau}\right)\ud Q_s,
\end{equation}
where the function $\tilde g:\R^+ \rightarrow \R$ is given by the formula
$$
\tilde g(t)=g(t)-e^{\int_0^t\lambda_s \ud s}\esp{\I_{\{\tau > t\}}g(\tau)},
$$
with
$$
g(x)=h(x \wedge T)\I_{\{x < T\}}.
$$
Here $h$ is the function introduced in \eqref{delta}. Moreover,
we have used the relationship
$$
\ud \tilde Q_t=-\frac{1}{1-F_t^\tau}\ud Q_t, \quad t \in [0,T],
$$
that follows from Lemma 5.1 of~\cite{br}.
Consequently, by applying It\^o's formula, we get
$$
\ud \hat U_{H}(t)=\hat p_{T,K}(t)\ud \Psi_t+\Psi_{t-}\ud \hat p_{T,K}(t),\quad t \in [0,T],
$$
and, thus
$$
\hat H=\hat p_{T,K}(0)\esp{g(\tau)}+\int_0^T\xi_s^{\hat H,0}\ud \hat P(s,T)+L_T^{\hat H},
$$
where the benchmarked risk-minimizing strategy is of the form
$$
\xi_t^{\hat H,0}=\Psi_{t-}\frac{\partial \hat p_{T,K}(t)}{\partial \hat P(t,T)}, \quad t \in [0,T],
$$
with $\xi_t^{\hat H,j}=0$, for $j \in \{1,2,\ldots,d\}$. Here the benchmarked cost appears as
$$
\hat C_t^\delta 
=\hat p_{T,K}(0)\esp{g(\tau)}+\int_{]0,t]}\hat p_{T,K}(s-)\left(\tilde g(s-)-\frac{1-F_T^\tau}{1-F_{s-}^\tau}\right)\ud Q_s,
$$
for every $t \in [0,T]$, see \eqref{psirepre}. Due to the boundness of the hedge ratio $\frac{\partial \hat p_{T,K}(t)}{\partial \hat P(t,T)}$, $\bar \delta$ and the process $\Psi$, the benchmarked hedgeable part of the contingent claim forms a square-integrable $\P$-martingale and the resulting strategy is $L^2$-admissible.

\begin{center}
{\bf Acknowledgement}
\end{center}

\noindent The authors like to thank Martin Schweizer and Wolfgang Runggaldier for valuable discussions on the manuscript. \\
The research leading to these results has received funding from the European
Research Council under the European Community's Seventh Framework Programme (FP7/2007-2013) / ERC grant agreement no [228087].

\appendix



\section{Technical Proofs} \label{tecpr}

\noindent Here we extend the result of Lemma 2.3 of~\cite{s01} to the case of a general discounting factor. Our proof is similar to the one of Lemma 2.3 of~\cite{s01}, however it contains some differences due to the fact that all the strategy's components contribute to the cost.\\

\noindent {\em Proof of Lemma \ref{lem:msf}.} 
Suppose $\delta$ is a benchmarked risk-minimizing strategy. Fix $t_0 \in [0,T]$ and define a strategy $\tilde \delta$ by setting for each $t \in [0,T]$
$$
\tilde \delta_t:=\delta_t\I_{[0,t_0)}(t)+\eta_t\I_{[t_0,T]}(t),
$$
where $\eta$ is an $\bF$-predictable process determined in a way such that the resulting strategy $\tilde \delta$ is $L^2$-admissible and 
$$
\hat S_t^{\tilde \delta}=\tilde \delta_t \cdot \hat S_t:=\hat S_t^{\delta}\I_{[0,t_0)}(t)+\condespf{\hat S_T^\delta - \int_t^T\delta_s\cdot \ud \hat S_s}\I_{[t_0,T]}(t).
$$
Here we assume to work with an RCLL version. Then $\tilde \delta$ is an $L^2$-admissible strategy with $\hat S_T^{\delta}=\hat S_T^{\tilde \delta}$ and
\begin{equation}\label{eq:cost}
\hat C_{t_0}^{\tilde \delta}=\mathbb E[\hat C_T^{\delta}|\F_{t_0}].
\end{equation}
Since $\hat C_T^\delta = \hat C_T^{\tilde \delta} + \int_{t_0}^T(\eta_u-\delta_u)\cdot \ud \hat S_u$, we have
$$
\hat C_T^\delta-\hat C_{t_0}^\delta=\hat C_T^{\tilde \delta}-\hat C_{t_0}^{\tilde \delta}+\mathbb E[\hat C_T^{\delta}|\F_{t_0}]-\hat C_{t_0}^{\delta} + \int_{t_0}^T(\eta_u-\delta_u)\cdot \ud \hat S_u.
$$
Taking the squares of both sides of the equation, we have
\begin{equation*}
\begin{split}
& \left(\hat C_T^\delta-\hat C_{t_0}^\delta\right)^2=\left(\hat C_T^{\tilde \delta}-\hat C_{t_0}^{\tilde \delta}\right)^2 + \left(\mathbb E[\hat C_T^{\delta}|\F_{t_0}]-\hat C_{t_0}^{\delta}\right)^2 + \left(\int_{t_0}^T(\eta_u-\delta_u)\cdot \ud \hat S_u\right)^2\\
& \qquad + 2\left(\hat C_T^{\tilde \delta}-\hat C_{t_0}^{\tilde \delta}\right)\left(\mathbb E[\hat C_T^{\delta}|\F_{t_0}]-\hat C_{t_0}^{\delta}\right)+2\left(\hat C_T^{\tilde \delta}-\hat C_{t_0}^{\tilde \delta}\right) \int_{t_0}^T(\eta_u-\delta_u)\cdot \ud \hat S_u\\
& \qquad \qquad +2\int_{t_0}^T(\eta_u-\delta_u)\cdot \ud \hat S_u\left(\mathbb E[\hat C_T^{\delta}|\F_{t_0}]-\hat C_{t_0}^{\delta}\right).
\end{split}
\end{equation*}
Then conditioning with respect to $\F_{t_0}$, by \eqref{eq:cost} we obtain
\begin{align*}
\hat R_{t_0}^\delta & = \hat R_{t_0}^{\tilde \delta} + \left(\mathbb E[\hat C_T^{\delta}|\F_{t_0}]-\hat C_{t_0}^{\delta}\right)^2 + \mathbb E \left[\left(\int_{t_0}^T(\eta_u-\delta_u)\cdot \ud \hat S_u\right)^2\Big{|}\F_{t_0}\right]\\
& \quad + 2\mathbb E \left[\left(\hat C_T^{\tilde \delta}-\hat C_{t_0}^{\tilde \delta}\right)\left(\mathbb E[\hat C_T^{\delta}|\F_{t_0}]-\hat C_{t_0}^{\delta}\right)\Big{|}\F_{t_0}\right]\\
& \quad \quad + 2\mathbb E \left[\left(\hat C_T^{\tilde \delta}-\hat C_{t_0}^{\tilde \delta}\right) \int_{t_0}^T(\eta_u-\delta_u)\cdot \ud \hat S_u\Big{|}\F_{t_0}\right]\\
&=\hat R_{t_0}^{\tilde \delta} + \left(\mathbb E[\hat C_T^{\delta}|\F_{t_0}]-\hat C_{t_0}^{\delta}\right)^2 + \mathbb E \left[\left(\int_{t_0}^T(\eta_u-\delta_u)\cdot \ud \hat S_u\right)^2\Big{|}\F_{t_0}\right]\\
& \quad + 2\left(\mathbb E[\hat C_T^{\delta}|\F_{t_0}]-\hat C_{t_0}^{\delta}\right)\cdot\mathbb E\Big[\underbrace{\hat C_T^{\tilde \delta}+\int_{t_0}^T(\eta_u-\delta_u)\cdot \ud \hat S_u}_{=\hat C_T^\delta}-\hat C_{t_0}^{\tilde \delta}\Big{|}\F_{t_0}\Big]\\
& \quad \quad + 2\mathbb E \left[\left(\hat C_T^{\tilde \delta}-\hat C_{t_0}^{\tilde \delta}\right) \int_{t_0}^T(\eta_u-\delta_u)\cdot \ud \hat S_u\Big{|}\F_{t_0}\right]\\
&=\hat R_{t_0}^{\tilde \delta} + \left(\mathbb E[\hat C_T^{\delta}|\F_{t_0}]-\hat C_{t_0}^{\delta}\right)^2+ \mathbb E \left[\left(\int_{t_0}^T(\eta_u-\delta_u)\cdot \ud \hat S_u\right)^2\Big{|}\F_{t_0}\right]\\
& \quad +  2\mathbb E \left[\left(\hat C_T^{\tilde \delta}-\hat C_{t_0}^{\tilde \delta}\right) \int_{t_0}^T(\eta_u-\delta_u)\cdot \ud \hat S_u\Big{|}\F_{t_0}\right].
\end{align*}
Because $\delta$ is benchmarked risk-minimizing, it has minimal risk.
If $\tilde \delta$ is also risk-minimizing, we must have
$$
\hat R_{t_0}^{\delta} = \hat R_{t_0}^{\tilde \delta}, \quad t \in [0,T]
$$
and
$$
\mathbb E \left[\left(\hat C_T^{\tilde \delta}-\hat C_{t_0}^{\tilde \delta}\right) \int_{t_0}^T(\eta_u-\delta_u)\cdot \ud \hat S_u\Big{|}\F_{t_0}\right]=0,
$$ 
since the residual optimal cost $\hat C_T^{\tilde \delta}-\hat C_{t_0}^{\tilde \delta}$ must be orthogonal to all integrals of the form $\int_{t_0}^T\xi_u \ud \hat S_u$, with $\xi$ $L^2$-admissible, by definition (i.e. by Definitions \ref{def:risk} and \ref{optimalstrategy}). Consequently, we obtain
$$
\left(\mathbb E[\hat C_T^{\delta}|\F_{t_0}]-\hat C_{t_0}^{\delta}\right)^2+ \mathbb E \left[\left(\int_{t_0}^T(\eta_u-\delta_u)\cdot \ud \hat S_u\right)^2\Big{|}\F_{t_0}\right]=0
$$
and we can easily conclude that
$$
\hat C_{t_0}^{\delta}=\mathbb E[\hat C_T^{\delta}|\F_{t_0}]\quad \P-{\rm a.s.}.
$$
Since $t_0$ is arbitrary, the assertion follows.

\begin{flushright}
$\square$
\end{flushright}

\section{Some Useful Definitions}

\noindent We recall briefly the definition of $\bF$-predictable projection of a measurable process endowed with some suitable integrability properties and the definition of $\bF$-predictable dual projection of a raw integrable increasing process.
\begin{theorem}[predictable projection]
Let $X$ be a measurable process either positive
or bounded. There exists an $\bF$-predictable process $Y$ such that
$$
\condespff{X_\tau \I_{\{\tau < \infty\}}}=Y_\tau \I_{\{\tau < \infty\}}\quad \P-{\rm a.s.}
$$
for every predictable $\bF$-stopping time $\tau$.
\end{theorem}

\begin{proof}
See~\cite{dm2} or~\cite{ry} for the proof. 
\end{proof}

\begin{definition}
Let $A$ be a raw integrable increasing process. The $\bF$-predictable dual projection of $A$ is the $\bF$-predictable increasing process $B$ defined by
$$
\esp{\int_{[0,\infty[}X_s \ud B_s}=\esp{\int_{[0,\infty[}{}^pX_s \ud A_s}.
$$
\end{definition}
\noindent For a further discussion on this issue, see e.g.~\cite{dm2}.


\end{document}